\documentclass[aps,superscriptaddress,twocolumn,twoside,floatfix,pra,a4paper]{revtex4-1}
\usepackage[colorlinks=true, citecolor=blue, urlcolor=red]{hyperref}
\usepackage[utf8]{inputenc}
\usepackage{graphicx}
\usepackage{subfigure}
\usepackage{setspace}
\usepackage{amsmath}
\usepackage{amsthm}
\usepackage{color}
\usepackage{color,soul}
\usepackage{appendix}
\usepackage{tgpagella}
\usepackage{amssymb}
\allowdisplaybreaks
\usepackage[dvipsnames]{xcolor}
\newtheorem{proposition}{Proposition}

\newtheorem{conj}{Conjecture}

\begin{document}
\title{Characterizing qubit channels in the context of quantum teleportation}

\author{Arkaprabha Ghosal}
\email{a.ghosal1993@gmail.com}
\affiliation{Centre for Astroparticle Physics and Space Science (CAPSS), 
Bose Institute, Block EN, Sector V, Salt Lake, Kolkata 700 091, India}

\author{Debarshi Das}
\email{dasdebarshi90@gmail.com}
\affiliation{S. N. Bose National Centre for Basic Sciences, Block JD, Sector III, Salt Lake, Kolkata 700 106, India}

\author{Subhashish Banerjee}
\email{subhashish@iitj.ac.in}
\affiliation{Indian Institute of Technology Jodhpur, Jodhpur 342 037, India}
\begin{abstract}
We consider a scenario where a party, say, Alice prepares a pure two-qubit (either maximally entangled or non-maximally entangled) state and sends one half of this state to another distant party, say, Bob through a qubit (either unital or non-unital) channel. Finally, the shared state is used as a teleportation channel. In this scenario, we focus on characterizing the set of qubit channels with respect to the final state's efficacy as a resource of quantum teleportation (QT) in terms of maximal average fidelity and fidelity deviation (fluctuation in fidelity values over the input states). Importantly, we point out the existence of a subset of qubit channels for which the final state becomes useful for universal QT (having maximal average fidelity strictly greater than the classical bound and having zero fidelity deviation) when the initially prepared state is either useful for universal QT (i.e., for a maximally entangled state) or not useful for universal QT (i.e., for a subset of non-maximally entangled pure states).  Interestingly, in the latter case, we show that non-unital channels (dissipative interactions) are more effective than unital channels (non-dissipative interactions) in producing useful states for universal QT from non-maximally entangled pure states.
\end{abstract}
\maketitle

\section{Introduction}
Developments in the understanding of quantum correlations  along with the superposition principle, have bench-marked advances in the field of quantum information.  Quantum entanglement is one of the most prominent of quantum correlations which is empowered due to the superposition principle. Entanglement plays a pivotal role in the success of a number of quantum information protocols, in particular, Quantum Teleportation (QT). 

QT \cite{teleb} can be realized as a strategy between two  spatially separated parties where a sender (say, Alice) transfers an unknown quantum state to the receiver (say, Bob) using Local Operations and Classical Communications (LOCC) and shared entanglement without any physical transmission of quantum systems. QT has played a crucial role in the
advancement of quantum communication. Motivated from QT, innumerable
quantum information theoretic and communication tasks ranging from quantum repeaters \cite{teleapp1}, quantum gate teleportation \cite{teleapp2} to measurement-based computing \cite{teleapp3} have been
proposed.  The idea of QT has been extended to multipartite systems \cite{mqt} and  continuous variable systems \cite{cqt}. Experimental demonstrations of   QT have been reported \cite{exp1,exp2,exp3} which also include QT over large distances \cite{exp4} or QT from ground to satellite  \cite{expgs}.

The standard figure of merit for QT is the average fidelity \cite{tele,can1,f1,f2}.  It represents  the average closeness or overlap  between the input state at the sender's  end and the output state at the receiver's end, where the average is taken over all possible input states. In case of  perfect QT, the output state is exactly equal to the input state for all possible input states. Maximally entangled states are the necessary resource for perfect QT \cite{teleb}. However, in reality one cannot expect maximally entangled states due to environmental interactions leading to imperfect QT, where the average fidelity is strictly less than one \cite{tele}.  In such cases, all input states may not be teleported equally well and   dispersion or fluctuation in fidelity over the input states may arise \cite{FD,fd1}. Although the average fidelity being the standard quantifier for QT, it does not contain any information about the fluctuation in fidelity or fidelity deviation \cite{fd1,FD,Opt,fd4,extra1,extra2}. Hence, average fidelity associated with fidelity deviation can completely characterize
QT.  The maximal average fidelity is the maximal value of average fidelity achievable over all possible local unitary operations within the standard  teleportation protocol \cite{tele}. On the other hand, fidelity deviation is something that one would like to minimize while keeping the average fidelity to the maximal possible value \cite{FD,Opt}.  Any two-qubit entangled state, for which the maximal fidelity is strictly greater than the classical bound, is known as an useful state for QT \cite{tele,can1}. A useful state for QT is called useful for Universal Quantum Teleportation (UQT) if and only if the state shows vanishing fidelity deviation \cite{FD}. In other words, a UQT implies that all input states can be teleported equally well with the same fidelity value equals to the maximal average fidelity. Hence, the concept of fidelity deviation can be used as a filter to select the optimal states for QT \cite{Opt}. 

Quantifying the success of QT based only an average fidelity has some  limitations in practical cases, for example, in the context of quantum circuit \cite{gate.dist, fd4,cite1} consisting of QT as an intermediate step. In such realistic scenarios, one may need to teleport a finite number of input states without considering all possible input states. In such cases, the knowledge of fidelity deviation plays a crucial role to estimate the individual fidelities for different input states. In particular, if the entangled channel is useful for UQT, then any input state is teleported with fidelity equal to the average fidelity.  


The necessary resource for QT is entanglement which must be shared between the sender and the receiver. Sharing of entanglement can be done by a simple process; Alice prepares an entangled pair in her lab and sends one half to Bob via a quantum channel. A perfect QT requires a maximally entangled state which can only be established via noiseless quantum channels.  In practical scenarios, the channels are noisy and, hence, these studies need to be considered by taking into account of the effects of environmental interactions. This can be done using the tools of open quantum systems \cite{breuerpet,sbbook}.

A quantum channel \cite{kraus,kraus2,choi} is a Completely Positive and Trace Preserving (CPTP) map $\Phi$ with the operator sum representation $\Lambda(\rho) = \sum_i K_i \rho K_i^{\dagger}$. Here $K_i$ are the Kraus  operators \cite{kraus,kraus2} obeying the completeness condition, $\sum_i K_i^{\dagger} K_i = I$. Any quantum  channel  is characterized  by  the  following  properties: (a) linearity, (b) Hermiticity preserving,  (c) positivity  preserving, and (d) trace preserving. The evolution modeled by a quantum channel could be, in general, unital or non-unital. 

 In the present study, we consider the scenario where Alice prepares either a maximally entangled two-qubit state or a non-maximally entangled two-qubit pure state and sends one half of it via a unital or non-unital qubit channel. The final state shared between Alice and Bob is used as the teleportation channel. Since the maximally entangled two-qubit states are useful for UQT \cite{FD}, in this case we find out the set of unital as well as non-unital channels for which the final state still remains useful for UQT. On the other hand, the non-maximally entangled two-qubit pure states being not useful for UQT \cite{FD}, we work out on finding the subset of unital and non-unital channels which convert the initial non-maximally entangled two-qubit pure states into useful states for UQT by acting on one half of the   states. Our results indicate that both unital and non-unital qubit channels can decrease the fidelity deviation (even can eliminate it completely).

Environmental interactions and the effects of quantum channels inevitably degrade the efficacy of a quantum resource. Hence, from foundational point of view as well as from information theoretic perspective it is important to find out the set of quantum channels that preserve the effectiveness of any quantum resource. QT being one of the primitive quantum information processing protocols, analyzing the set of quantum channels preserving the resources for QT is of paramount significance. Our present study is motivated to address this practical issue. Most importantly, our results effectively filter out the set of qubit channels which can be used in practical scenario for realizing UQT.

 The paper is organized as follows. Sec. \ref{sec2} is dedicated to the preliminary ideas and definitions useful for our paper.  In particular, we briefly discuss  the concept of concurrence for two-qubit states, the Hilbert-Schmidt representation and canonical representation of an arbitrary two-qubit state, the concept of maximal average fidelity and fidelity deviation for a two-qubit state, and the  qubit channels. In Sec. \ref{sec3}, we summarize the results obtained in this paper. Next, in Sec. \ref{sec4}, we present the results when the initial state is a Bell state followed by the results  in Sec. \ref{sec5} with non-maximally entangled two-qubit pure states as the initial states. An analysis of physically motivated noise models is made in Sec. \ref{sec6}. Finally, we conclude with a brief discussion in Sec. \ref{sec7}. Some of the technical details are relegated to two Appendices.

\section{Preliminaries} \label{sec2}
In this section, we will discuss the basic definitions and preliminaries of  quantum entanglement, Hilbert-Schmidt representation and canonical representation of a two-qubit state, maximal average fidelity and fidelity deviation in QT with a two-qubit state and qubit channels.

\subsection{Concurrence of a two-qubit state}

 Entanglement is a fundamental aspect of quantum correlation present in compound quantum systems. 
There exist a number of well-known measures of quantum entanglement. In the present paper, we restrict ourselves to the concurrence measure \cite{conc}. For a two qubit state $\rho $, the concurrence $C(\rho )$ is defined as \cite{conc},
\begin{equation}
C(\rho)= \max \lbrace 0,\lambda_{1}-\lambda_{2}-\lambda_{3}-\lambda_{4} \rbrace.\label{C1}
\end{equation} 
Here $\lambda_1 \geq \lambda_2 \geq \lambda_3 \geq \lambda_4$ are square roots of the eigenvalues of $\rho \tilde{\rho}$, where $\tilde{\rho}=(\sigma_{y}\otimes \sigma_{y})\rho^{*} (\sigma_{y}\otimes \sigma_{y})$ is the Pauli rotated state with $\sigma_y$ being the Pauli bit-phase flip matrix and $\rho^{*}$ being the complex conjugation of $\rho$ in the computational basis. 
\paragraph*{}
\subsection{Hilbert-Schmidt representation and the canonical form of a two-qubit state}
The Hilbert-Schmidt representation of a two-qubit density matrix $\rho$ is given by \cite{tele,hs1,hs2},
\begin{equation}
\rho = \frac{1}{4} \left[ \mathbf{I}_{4}+ \boldsymbol{R} \cdot \boldsymbol{\sigma} \otimes \mathbf{I}_{2} +\mathbf{I}_{2} \otimes \boldsymbol{S} \cdot \boldsymbol{\sigma} +\sum_{i,j=1}^{3}T_{ij}\sigma_{i}\otimes \sigma_{j}\right]. \label{HS}
\end{equation} 
The terms $\boldsymbol{R}$, $\boldsymbol{S}$ represent local vectors in $\mathcal{R}^{3}$ in each respective marginal   and $\boldsymbol{R} (\boldsymbol{S}) \cdot \boldsymbol{\sigma}$ = $\sum_{i=1}^{3} R_i (S_i) \sigma_i$ with $\sigma_i$ ($i=1,2,3$) being the Pauli matrix; $T_{ij}=\text{Tr}(\rho \sigma_{i}\otimes \sigma_{j})$ are the elements of the $3\times 3$ correlation matrix $T_{\rho}$, where $i,j=1,2,3$. Let $t_{11}$, $t_{22}$, $t_{33}$ are the eigenvalues of $T_{\rho}$. Then there always exits a product unitary operation $U_{1}\otimes U_{2}$ that transforms $\rho \rightarrow \rho_{C}$ such that \cite{hs1,hs2,can1}
\begin{eqnarray}
 \rho_{C} &=& (U_{1}\otimes U_{2})\rho (U_{1}\otimes U_{2})^{\dagger} \nonumber \\
 &=&\frac{1}{4} \left[ \mathbf{I}_{4}+\boldsymbol{r} \cdot \boldsymbol{\sigma} \otimes \mathbf{I}_{2} + \mathbf{I}_{2}\otimes \boldsymbol{s} \cdot \boldsymbol{\sigma} + \sum_{k=1}^{3} \lambda_k |t_{kk}| \, \sigma_{k} \otimes \sigma_{k}\right], \nonumber \\ \label{HSD}   \end{eqnarray}
with $\lambda_{k}\in \lbrace -1,+1\rbrace$ for $k=1,2,3$; $\boldsymbol{r}$, $\boldsymbol{s}$ represent local vectors in $\mathcal{R}^{3}$ in each respective marginal  and $\boldsymbol{r} (\boldsymbol{s}) \cdot \boldsymbol{\sigma}$ = $\sum_{i=1}^{3} r_i (s_i) \sigma_i$. Now, one can further choose the product unitary $U_{1}\otimes U_{2}$ such that (i) if det$(T_{\rho})\leq 0$, then $\lambda_k$ = $-1$ for $|t_{kk}| \neq 0$, $k=1,2,3$; (ii) if det$(T_{\rho}) > 0$, then $\lambda_i, \lambda_j = -1$, $\lambda_k = +1$ for $|t_{kk}| \neq 0$ for any choice of $i \neq j \neq k \in \{1,2,3\}$ with $|t_{ii}| \geq |t_{jj}| \geq |t_{kk}|$. This transformed $\rho_C$ is known as the canonical form of $\rho$ \cite{can2,FD}.

\subsection{Maximal average fidelity and fidelity deviation in QT with a two-qubit state}

Perfect QT is achieved if and only if the shared state is maximally entangled. In this case, the output state of QT is exactly equal to the input state. However, in practice, the available states are mixed entangled and hence,  QT will not be perfect. The standard figure of merit for QT is expressed through the concept of  average fidelity \cite{tele,can1,f1,f2}, which signifies the closeness between the input and the output states.

The average teleportation fidelity for a two-qubit state $\rho$ is defined as  \cite{tele}
\begin{equation}
\langle f_{\rho} \rangle =\int f_{\psi, \rho} d\psi,
     \label{av.f}
\end{equation}
where $f_{\psi , \rho}=\langle \psi |\chi |\psi \rangle$  is the fidelity between the input-output pair $(|\psi \rangle \langle \psi |, \chi)$. The above integration is taken over a uniform distribution of all possible pure qubit input states $|\psi\rangle$ (normalized Haar measure, $\int d \psi =1$).  In other words, this integration is over the parameters appearing in $|\psi\rangle$. In Bloch representation, an arbitrary pure qubit input state is given by $|\psi \rangle \langle \psi| = \frac{1}{2}(\mathbf{I}_2 + \boldsymbol{\hat{a}} \cdot \boldsymbol{\sigma})$, where the unit vector $\boldsymbol{\hat{a}}$ represents the Bloch vector of the input state and it is given by,  $\boldsymbol{\hat{a}} = (\sin \theta \cos \phi, \sin \theta \sin \phi, \cos \theta)$. With such parametrization of an arbitrary input state, we have $d \psi = \sin \theta d \theta d \phi$. Note that the average fidelity $\langle f_{\rho} \rangle$ is defined when the standard protocol for QT proposed by Bennett \textit{et al.} \cite{teleb} is adopted. As mentioned earlier, $\langle f_{\rho} \rangle = 1$ is possible if and only if $\rho$ is maximally entangled. 
\paragraph*{}
Fidelity deviation $\delta_{\rho}$ is a secondary quantifier of QT which measures fluctuations in fidelity  over the input states. It is defined as the standard deviation of fidelity values over all possible input states given by \cite{fd1,FD,Opt,fd4},
\begin{equation}
    \delta_{\rho}=\sqrt{\langle f_{\rho}^{2}\rangle -\langle f_{\rho}\rangle ^{2}}, \label{dev}
\end{equation} 
where $\langle f^2_{\rho} \rangle =\int f^2_{\psi, \rho} d\psi$, and $0\leq \delta_{\rho}\leq \frac{1}{2}$.

For a given two-qubit state $\rho$, the maximal average fidelity (or, maximal fidelity) $F_{\rho}$ is defined as the maximal value of average fidelity obtained over all strategies under the standard protocol and local unitary operations \cite{tele,can2},
\begin{equation}
   F_{\rho}=\max _{\text{LU}}\langle f_{\rho}\rangle,
\end{equation} 
where the maximization is done over all possible local unitary (LU) strategies. The protocol which maximizes the average fidelity is known to be the `optimal protocol'. Now, one can show that \cite{can2}
\begin{equation}
F_{\rho}= \langle f_{\rho_C}\rangle.
\label{optcan}
\end{equation}
The above Eq.(\ref{optcan}) indicates that an optimal protocol consists of two steps: (i) transforming $\rho \rightarrow \rho_{C}$ using an appropriate local unitary operation, and then (ii) using $\rho_C$ for QT following the standard protocol proposed by Bennett \textit{et al.} \cite{teleb}.

Since the primary motivation of QT is to maximize the average fidelity, fidelity deviation should be analyzed for optimal protocol. The fidelity deviation corresponding to the optimal protocol can be defined as \cite{FD,Opt}
\begin{equation}
\Delta_{\rho}=\delta_{\rho_C}.
\end{equation} 

A two-qubit state $\rho$ is useful for QT iff $F_{\rho}>\frac{2}{3}$ \cite{tele,can1}, where $\frac{2}{3}$ is the maximum average fidelity obtained in classical protocols. On the other hand, a two-qubit state $\rho$ is universal for QT iff $\Delta_{\rho}=0$ \cite{FD}. If a two-qubit state satisfies the above universality condition, then all input states will be teleported with the same fidelity. 

It has been shown earlier \cite{hs2,can2} that useful states for QT form a subset of the states with the property det$(T_{\rho}) < 0$. The analytical expressions of maximal fidelity and fidelity deviation for two-qubit states with det$(T_{\rho}) < 0$ can be written as \cite{can2,FD}
\begin{eqnarray}
F_{\rho}&=& \frac{1}{2}\left( 1+\frac{1}{3}\sum\limits_{i=1}^{3}|t_{ii}|\right), \nonumber \\
\Delta_{\rho}&=& \frac{1}{3\sqrt{10}}\sqrt{\sum_{i < j =1}^{3}(|t_{ii}|-|t_{jj}|)^{2}}.
\label{mf&dev}
\end{eqnarray}
From the above equations, it follows that a two-qubit state $\rho$ is useful  for Universal Quantum Teleportation (UQT) (i.e., useful and universal for QT) if and only if  $|t_{11}|=|t_{22}|=|t_{33}|>\frac{1}{3}$ \cite{FD}.

Hence, in order to theoretically determine usefulness and universality of a two-qubit state $\rho$ in the context of QT, we only need to find out the eigenvalues of the correlation matrix $T_{\rho}$, we don't need to find out the optimal protocol or the canonical form $\rho_C$.

\subsection{Qubit channels}

If any qubit state $\chi$ is passed through a channel $\Lambda$, the output state $\chi_{\Lambda}$ can be written as \cite{kraus,kraus2,choi},   
\begin{equation}
    \Lambda(\chi)=\chi_{\Lambda}=\sum_{i=0}^{r_{\Lambda}-1}K^{\Lambda}_{i}\chi K_{i}^{\Lambda^{\dagger}}.\label{qchannel}
\end{equation}
This type of representation is known as the operator sum representation or the Kraus representation \cite{kraus,kraus2} where $\lbrace K^{\Lambda}_{i}\rbrace$ are known as the Kraus operators. The Kraus operators always satisfy the completeness property given by,
\begin{equation}
\sum_{i=0}^{r_{\Lambda}-1} K_{i}^{\Lambda^{\dagger}} K^{\Lambda}_{i}=\mathbf{I}_2 , \label{complt}
\end{equation}
where the quantity $r_{\Lambda}$ represents the number of Kraus operators. 

In general, there is no unique representation of the Kraus operators 
 corresponding to a particular qubit channel \cite{kraus,kraus2,QC}. For a qubit channel $\Lambda$, one can find another Kraus representations $\lbrace \tilde{K}^{\Lambda}_{i} \rbrace $ related with $\lbrace K^{\Lambda}_{i}\rbrace$ by the relation \cite{omkarsinglequbit},
\begin{equation}
    \tilde{K}^{\Lambda}_{i}=\sum_{j}W_{ij}K^{\Lambda}_{j}, \label{Krr}
\end{equation}
where $W_{ij}$ is any unitary transformation such that $W_{ij}^{\dagger} W_{ij}$= $W_{ij} W_{ij}^{\dagger}$ = $\mathbf{I}_2$. Hence, there exist infinite number of possible Kraus representations for any given qubit channel $\Lambda$.

 Any qubit channel $\Lambda $ can be categorized in two classes- (i) unital class of channels $\Lambda_{u}$ and (ii) non-unital class of channel $\Lambda_{nu}$. Unital channels always preserve the identity operator, i.e., $\Lambda_{u} (\mathbf{I}_{2})=\mathbf{I}_{2}$. Whereas, for any non-unital channel $\Lambda_{nu}$ one has $\Lambda (\mathbf{I}_{2})\neq \mathbf{I}_{2}$.  Consequently, for any unital map $\Lambda_{u}$, the equality $\sum_{i}K^{\Lambda_u}_{i}K_{i}^{\Lambda_u^\dagger}=\mathbf{I}_{2}$ always holds. However, for any non-unital map $\Lambda_{nu}$, the above equality does not hold, i.e., $\sum_{i}K^{\Lambda_{nu}}_{i}K_{i}^{\Lambda_{nu}^\dagger} \neq \mathbf{I}_{2}$. 

It can be shown that any convex combination of any two unital qubit channels also represents a unital qubit channel \cite{ruskai}.  Hence, they form a convex set with four possible extreme points which are the four Pauli channels. In other words, any unital qubit channel can be expressed as a convex combination of four Pauli channels. The Kraus operators associated with the extreme points are $\{U \sigma_0 V, U \sigma_1 V, U \sigma_2 V, U \sigma_3 V \}$ \cite{ruskai}, where $U$ and $V$ are unitaries. The action of a unital qubit channel $\Lambda_u$ on the state $\chi$  can be expressed as \cite{QC,ruskai}
\begin{align}
\Lambda_u(\chi)&=\chi_{\Lambda_u} \nonumber \\
&=U \left(\sum_{i=0}^{3} p_i \sigma_i \left(V \chi V^{\dagger} \right) \sigma_i \right) U^{\dagger}, \quad 0 \leq p_i \leq 1 \, \forall \, i, .
\end{align}
with $\sum_{i=0}^{3} p_i =1$. 

When a qubit channel is realized with only one Kraus operator, then that channel must be unital \cite{ruskai}. In this case, the only Kraus operator will be unitary and such a channel is called a unitary channel.

In general, the number of Kraus operators $r_{\Lambda}$ has no specific upper bound. However, the lower bound on the number of operators, i.e., $r_{\Lambda}^{\text{min}}$  becomes important while representing any channel. The minimum number of Kraus operators for any given channel can be understood from the concept of Choi states. Let us consider a bipartite scenario where Alice prepares the Bell state $|\Phi_{1}\rangle$ and sends one half to Bob via any $\Lambda$. Here $|\Phi_1 \rangle$ is one of the states in Bell basis given by
\begin{align}
|\Phi_1 \rangle &= \frac{1}{\sqrt{2}} \left( |00 \rangle + |11 \rangle \right), \quad |\Phi_2 \rangle = \frac{1}{\sqrt{2}} \left( |01 \rangle + |10 \rangle \right),\nonumber \\
|\Phi_3 \rangle &= \frac{1}{\sqrt{2}} \left( |01 \rangle - |10 \rangle \right), \quad |\Phi_4 \rangle = \frac{1}{\sqrt{2}} \left( |00 \rangle - |11 \rangle \right).\label{bellstates}
\end{align}
The final state $\rho_{\Lambda,\Phi_{1}}$ shared between Alice and Bob after the channel interaction is known as the Choi state or dual state of $\Lambda$ given by \cite{choi,jam,choi2},
\begin{eqnarray}
\rho_{\Lambda,\Phi_{1}}&=&(\mathbf{I}\otimes \Lambda)|\Phi_{1}\rangle \langle \Phi_{1}| \nonumber \\
&=& \Phi_{\Lambda}~|\Phi_{1}\rangle \langle \Phi_{1}|.
\label{Choimap} 
\end{eqnarray}
A qubit channel $\Lambda$ is completely positive, if and only if its Choi state $\rho_{\Lambda,\Phi_{1}}$ is non-negative \cite{choi}. The trace-preserving condition of $\Lambda$ implies that  the marginal of Alice for $\rho_{\Lambda, \Phi_{1}}$ is always maximally mixed, i.e., $\text{Tr}_{2}(\rho_{\Lambda ,\Phi_{1}})=\mathbf{I}_{2}$.  Eq.(\ref{Choimap}) represents the Choi-Jamiolkowski isomorphism \cite{choi,jam} between a channel $\Lambda$ and its dual state $\rho_{\Lambda ,\Phi_{1}}$. Hence, it is obvious that the inherent geometry of the state $\rho_{\Lambda ,\Phi_{1}}\in \mathcal{L}(\mathbb{C}^{2}\otimes \mathbb{C}^{2})$ will be similar with the geometry of $\mathbf{I}\otimes \Lambda \in \mathcal{L}(\mathbb{C}^{2}\otimes \mathbb{C}^{2})$. 

 Let us describe the Hilbert-Schmidt decomposition of a Choi state. Up to unitary rotations, the state $\rho_{\Lambda,\Phi_{1}}$ can be written in the following canonical form \cite{QC,QC2,QC3},
\begin{eqnarray}
\rho_{\Lambda,\Phi_{1}}=\frac{1}{4}\left[ \mathbf{I}_{4}+\mathbf{I}_{2}\otimes \boldsymbol{s} \cdot \boldsymbol{\sigma} +\sum_{k}\lambda_{k}|t_{kk}| \, \sigma_{k}\otimes \sigma_{k}\right], \nonumber \\
\label{choistate}
\end{eqnarray}
where $\boldsymbol{s}$ $\equiv$ $(s_1, s_2, s_3)$ is the local vector at Bob's side and $\lambda_{k}|t_{kk}|$ are the eigenvalues of the correlation matrix $T_{\rho_{\Lambda,\Phi_{1}}}$  such that $\lambda_{k}\in \lbrace -1,+1 \rbrace$. The rank of the Choi state $\rho_{\Lambda,\Phi_{1}}$ is given by the the minimal number $r_{\Lambda}^{\text{min}}$ of Kraus operators associated with the channel $\Lambda$. For qubit channels, we have $1 \leq r_{\Lambda}^{\text{min}} \leq 4$.

For a unital channel $\Lambda_u$, up to local unitary rotations, the Choi state has the following Bell-diagonal form \cite{oneshot,ruskai}, 
\begin{eqnarray}
\rho_{\Lambda_u,\Phi_{1}}=\frac{1}{4}\left[ \mathbf{I}_{4}+\sum_{k}\lambda_{k}|t_{kk}| \,\sigma_{k}\otimes \sigma_{k}\right], \nonumber \\ 
\label{unitalchoistate}
\end{eqnarray}
which implies that $\text{Tr}_{1}(\rho_{\Lambda_u ,\Phi_{1}})=\mathbf{I}_{2}$ for any unital channel $\Lambda_u$. These Choi states form a convex set with maximally entangled states being the extreme points.
On the other hand, for a non-unital channel $\Lambda_{nu}$, up to unitary rotations, the Choi state $\rho_{\Lambda_{nu},\Phi_{1}}$ has the form given by Eq.(\ref{choistate}) with $\boldsymbol{s} \neq (0, 0, 0)$, i.e.,  $\text{Tr}_{1}(\rho_{\Lambda_{nu} ,\Phi_{1}}) \neq \mathbf{I}_{2}$.

Next, we will summarize the results obtained in this paper followed by detailed proofs and analysis of the results.

\section{Summary of the results} \label{sec3}
In the present study we consider two scenarios.
\paragraph*{}
\paragraph*{\textbf{Scenario 1:}} \textit{Alice prepares a maximally entangled two-qubit state and sends one half to Bob through a qubit channel.}

\paragraph*{}
Here the initial state is useful for UQT \cite{FD} and our aim is to find out the set of quantum channels for which the final state remains useful for UQT. The channel interaction can either be dissipative (non-unital) or non-dissipative (unital). In case of unital channels, we show that the final state is useful for UQT if and only if the channel is unitary for a particular single parameter channel associated with a rank-four Choi state. On the other hand, for non-unital channels, we find out that the final state remains useful for UQT if and only if the channel belongs to a strict subset  associated with rank-three and rank-four Choi states. In this case, we also derive the most general form of orthogonal Kraus operators associated with the non-unital channels that preserve the usefulness and universality.

Next, we consider another scenario: 
\paragraph*{}
\paragraph*{\textbf{Scenario 2:}} \textit{Alice prepares a pure non-maximally entangled two-qubit state and sends one half of it to Bob through a qubit channel.}

\paragraph*{}

In this case, the initial state is not useful for UQT (this state is useful for QT, but has non-vanishing fluctuation in fidelity) \cite{FD}. We want to find out whether there exists any quantum channel for which the final state becomes useful for UQT. When the channel is unital, we show that the final state is useful for UQT for a strict subset of unital channels if and only if the concurrence of the initial state is strictly greater than $\frac{1}{2}$. On the other hand, we demonstrate that the final state is useful for UQT for a strict subset of non-unital channels when the concurrence of the initial state is strictly greater than a critical value. The critical value in this case is less than $\dfrac{1}{2}$, implying an advantage of non-unital channels over unital interactions. Hence, these results indicate that the local interaction of unital as well as non-unital channels can eliminate the fluctuation in fidelity values.  Moreover, these results also point out that non-unital interactions are more effective than unital interactions in producing desirable states for UQT.

Finally, we supplement our studies with some quantum channels motivated from physical noise models.

\section{ Bell state as the initial state} \label{sec4}

Here we consider the  scenario 1, where Alice prepares a two-qubit Bell state $|\Phi_{1}\rangle = \frac{1}{\sqrt{2}}\left( |00\rangle +|11\rangle \right)$ and sends half of this state to Bob through a qubit channel $\Lambda$. Hence, in this case, the initially prepared state is useful  for UQT \cite{FD}. Here our goal is to find out the class of qubit channels for which the final state will also be useful  for UQT. That is, we want to characterize the qubit channels that preserve usefulness and universality in the context of QT. 

We will start by analyzing the rank of Choi states associated with the qubit channels $\Lambda$ that preserve usefulness and universality. 

\begin{proposition}
If Alice sends one half of a Bell state $|\Phi_1\rangle$ through any qubit channel $\Lambda^1$ associated with rank-one Choi state, then the final shared state will always be useful for UQT.
\label{prop1}
\end{proposition}

\begin{proof}
Any quantum channel $\Lambda^1$ associated with rank-one Choi state can be implemented with only one Kraus operator and that Kraus operator must be a unitary operator \cite{ruskai}. Hence, sending one half of a Bell state $|\Phi_1\rangle$ through any quantum channel $\Lambda^1$ associated with rank-one Choi state is equivalent to applying local unitary operation on the Bell state. The final shared state, therefore, will be a maximally entangled state, which is useful for UQT \cite{FD}.
\end{proof}

\begin{proposition}
If Alice sends one half of a Bell state $|\Phi_1\rangle$ through any qubit channel $\Lambda^2$ associated with rank-two Choi state, then the final shared state will never be useful for UQT.
\label{prop2}
\end{proposition}

\begin{proof}
This proof mainly follows from results presented in \cite{QC}. Let Alice prepares a Bell state $|\Phi_1\rangle$ and sends one half of that state through any qubit channel $\Lambda^2$. Then the finally shared two-qubit state $\rho^f$ is nothing but the Choi state $\rho_{\Lambda^2, \Phi_{1}}$ of the channel, i.e.,
\begin{widetext}
\begin{align}
  \rho^f =  \rho_{\Lambda^2, \Phi_{1}} 
  =\frac{1}{4}\left( \begin{array}{cccc}
1+\tilde{s}_{3}+\tilde{t}_{33} & \tilde{s}_{1}-i\tilde{s}_{2} & 0 & \tilde{t}_{11}-\tilde{t}_{22} \\
\tilde{s}_{1}+i\tilde{s}_{2} & 1-\tilde{s}_{3}-\tilde{t}_{33} & \tilde{t}_{11}+\tilde{t}_{22} & 0 \\ 
0 & \tilde{t}_{11}+\tilde{t}_{22} & 1+\tilde{s}_{3}-\tilde{t}_{33} & \tilde{s}_{1}-i\tilde{s}_{2} \\
\tilde{t}_{11}-\tilde{t}_{22} & 0 & \tilde{s}_{1}+i\tilde{s}_{2} & 1-\tilde{s}_{3}+\tilde{t}_{33}
\end{array}
\right),
\label{ChoiM}
\end{align}
\end{widetext}
where the above matrix is written in the computational basis $\{|00\rangle, |01\rangle, |10\rangle, |11 \rangle\}$. Henceforth, all $4 \times 4$ matrices will be written in this basis.

Since the channel $\Lambda^2$ is associated with rank-two Choi state, rank of $\rho^f$ will be two. Hence, linear combinations of $3\times 3$ minors of $\rho^f$ must be zero, which implies the following three conditions \cite{QC},
\begin{eqnarray}
\tilde{s}_{3}\left( \tilde{t}_{33}+\tilde{t}_{11}\tilde{t}_{22}\right)&=&0,  \nonumber \\
\tilde{s}_{2}\left( \tilde{t}_{22}+\tilde{t}_{11}\tilde{t}_{33}\right)&=&0,  \nonumber \\
\tilde{s}_{1}\left( \tilde{t}_{11}+\tilde{t}_{22}\tilde{t}_{33}\right)&=&0.
\label{rank2cnd}
\end{eqnarray}
These conditions, together with the fact that diagonal elements of a positive semi-definite matrix are always greater than the elements in the same column, lead to the conclusion that all $\tilde{s}_{k}$ but one have to be equal to zero if $\rho^f$ is rank-two \cite{QC}. Without loss of generality, one can choose $\tilde{s}_{1}$ = $\tilde{s}_{2}$ = $0$ and parameterize $\tilde{t}_{11}  = \cos \alpha$, $\tilde{t}_{22}  = \cos \beta$. Hence, we have $\tilde{t}_{33} = -\cos \alpha \cos \beta$. 
The state $\rho^f$ will be useful for UQT if and only if $|\tilde{t}_{11}|$ = $|\tilde{t}_{22}|$ = $|\tilde{t}_{33}|$ $>$ $\frac{1}{3}$, where $\tilde{t}_{11} = \cos \alpha$, $\tilde{t}_{22} = \cos \beta$, $\tilde{t}_{33} = -\cos \alpha \cos 
\beta$. This will be satisfied if and only if $|\cos \alpha|   = |\cos \beta| = 1$. But this implies that the state $\rho^f$ is rank-one. Hence proved.
\end{proof}

Next, we provide an example which supports Proposition \ref{prop2}. Let us consider the dephasing channel $\Lambda_{\text{dephasing}}$ with the following Kraus operators: $K^{\Lambda_{\text{dephasing}}}_0 = \sqrt{p}~ \mathbb{I}$, $K^{\Lambda_{\text{dephasing}}}_1= \sqrt{1-p} ~\sigma_{3}$ with $0< p < 1$. Note that this channel is unital as $\sum\limits_{i=0}^{1} \left(K^{\Lambda_{\text{dephasing}}}_i\right) \left(K^{\Lambda_{\text{dephasing}}}_i\right)^{\dagger}= \mathbf{I}_2$. The final shared state between Alice and Bob in this case is given by,
\begin{align}
\rho^f_{\text{dephasing}} & = \sum_{i=0}^{1} (\mathbf{I}\otimes K^{\Lambda_{\text{dephasing}}}_{i}) |\Phi_1 \rangle \langle \Phi_1 | (\mathbf{I}\otimes K_{i}^{\Lambda_{\text{dephasing}}^{\dagger}}) \nonumber \\
&=p|\Phi_{1}\rangle \langle \Phi_{1}|+(1-p)|\Phi_{4}\rangle \langle \Phi_{4}|.
\end{align}
The maximal fidelity and fidelity deviation for $\rho^f_{\text{dephasing}}$ is given by,
\begin{eqnarray}
F_{\rho^f_{\text{dephasing}}}&=& \left\lbrace \begin{array}{cc}
    \dfrac{2p+1}{3}>\dfrac{2}{3} & \quad \text{when} \, \, \dfrac{1}{2}< p < 1, \\
    \\
    \dfrac{2}{3} & \quad \text{when} \, \, p= \dfrac{1}{2},\\
    \\
    1-\dfrac{2p}{3}>\dfrac{2}{3} & \quad \text{when} \, \, 0< p < \dfrac{1}{2},
\end{array}  
\right. \nonumber \\
\nonumber \\
\Delta_{\rho^f_{\text{dephasing}}}&=& \left\lbrace \begin{array}{cc}
    \dfrac{2(1-p)}{3\sqrt{5}}\neq 0 & \quad \text{when} \, \, \dfrac{1}{2}< p < 1, \\
    \\
    \dfrac{1}{3\sqrt{5}}\neq 0 & \quad \text{when} \, \,  p = \dfrac{1}{2},\\
    \\
    \dfrac{2p}{3\sqrt{5}}\neq 0 & \quad \text{when} \, \, 0< p < \dfrac{1}{2}.
\end{array}  
\right. \nonumber
\end{eqnarray}
Hence, the final shared state in this case is useful, but not universal for QT.

\begin{proposition}
If Alice sends one half of a Bell state $|\Phi_1\rangle$ through any qubit channel $\Lambda^3$ associated with rank-three Choi state, then the final shared state will be useful for UQT when $\Lambda^3$ belongs to a strict subset of all qubit channels  associated with rank-three Choi states.
\label{prop3}
\end{proposition}
\begin{proof}
We will prove this proposition by presenting two examples. Let us consider the qubit channel $\Lambda^3$ associated with a rank three Choi state with the following Kraus operators,
\begin{eqnarray}
&& K^{\Lambda^3}_{0}=\left( 
\begin{array}{cc}
\sqrt{1-p} & 0 \\
0 & 0
\end{array}
\right), \nonumber \\
&&K^{\Lambda^3}_{1}=\left( 
\begin{array}{cc}
0 & \sqrt{1-p} \\
0 & 0
\end{array}
\right), \nonumber \\
&& K^{\Lambda^3}_{2}=\left( 
\begin{array}{cc}
\sqrt{p} & 0 \\
0 & \sqrt{p}
\end{array}
\right),
\label{nonkrausex}
\end{eqnarray}
with $0 < p < 1$. The above matrices are written in the basis $\{|0\rangle, |1\rangle \}$. Henceforth, all $2 \times 2$ matrices will be expressed in this basis. Here one should note that  $\sum_{i} (K^{\Lambda^3}_{i}) (K^{\Lambda^3}_{i})^{\dagger}\neq \mathbf{I}$, which implies that the above channel is non-unital. In this case, the final shared state between Alice and Bob is given by,
\begin{eqnarray}
\rho^f_{\Lambda^3}&=&\sum_{i=0}^{2}(\mathbf{I}\otimes K^{\Lambda^3}_{i})|\Phi_{1}\rangle \langle \Phi_{1}|(\mathbf{I}\otimes K^{\Lambda^{3^\dagger}}_{i}) \nonumber \\
&=& p|\Phi_{1}\rangle \langle \Phi_{1}|+(1-p)\frac{\mathbf{I}_{2}}{2}\otimes |0\rangle \langle 0|.
\end{eqnarray}
One can easily check that the maximal average fidelity $F_{\rho^f_{\Lambda^3}}$ and fidelity deviation $\Delta_{\rho^f_{\Lambda^3}}$ of the above state are given by, 
\begin{eqnarray}
F_{\rho^f_{\Lambda^3}}&=&\frac{1+p}{2} \nonumber \\
&>&\frac{2}{3} \quad \text{when} \, \, \frac{1}{3}<p<1, \nonumber \\
\Delta_{\rho^f_{\Lambda^3}}&=& 0 \quad ~~ \forall \quad p \in (0,1). \nonumber
\end{eqnarray}
Hence, in this case, the final shared state is useful for UQT for a particular range of the channel parameter. This particular example shows that in the whole set of qubit channels associated with rank-three Choi states, for a subset of channels the final state will be useful for UQT. Next, we will now show that this subset is strict by giving another example where the final state is useful but not universal. 

Let us consider the qubit (unital) channel $\tilde{\Lambda}^3$ associated with a rank three Choi state with the following Kraus operators, 
\begin{align}
K^{\tilde{\Lambda}^3}_{0}=\sqrt{p_{0}}~\mathbf{I},~~K^{\tilde{\Lambda}^3}_{1}=\sqrt{p_{1}}~\sigma_{1}, ~~K^{\tilde{\Lambda}^3}_{2}=\sqrt{p_{2}}~\sigma_{2}
\end{align}
where $\sum\limits_{i=0}^{2} p_{i}=1$ and $1>p_0 \geq p_1 \geq p_2>0$. The final shared state between Alice and Bob in this case is given by,
\begin{align}
    \rho^f_{\tilde{\Lambda}^3}=&\sum_{i=0}^{2}p_{i}|\Phi_{i+1}\rangle \langle \Phi_{i+1}|.
\end{align}
When $0 < p_0 \leq \frac{1}{2}$, the above final state is not entangled and, hence, is not useful for UQT \cite{FD}. When $\frac{1}{2} < p_0 < 1$, the maximal average fidelity and fidelity deviation for this state are given by, 
\begin{eqnarray}
F_{\rho^f_{\tilde{\Lambda}^3}}&=&\frac{2p_{0}+1}{3}
> \frac{2}{3}, \nonumber \\
\Delta_{\rho^f_{\tilde{\Lambda}^3}}&\neq & 0. \nonumber 
\end{eqnarray}
Hence, in this case, the final state is never useful for UQT.
\end{proof}

\begin{proposition}
If Alice sends one half of a Bell state $|\Phi_1\rangle$ through any qubit channel $\Lambda^4$ associated with rank-four Choi state, then the final shared state will be useful for UQT when $\Lambda^4$ belongs to a strict subset of all qubit channels  associated with rank-four Choi states.
\label{prop4}
\end{proposition}
\begin{proof}
Here also, we will prove the proposition by presenting two examples. Consider the qubit (unital) channel $\Lambda^4$ associated with a rank four Choi state having Kraus operators given by,
\begin{equation}
K^{\Lambda^4}_{0}=\sqrt{p}\,\mathbf{I},~K^{\Lambda^4}_{i}=\sqrt{\frac{1-p}{3}}\sigma_{i},~~i=1,2,3,    
\end{equation}
with $0< p < 1$. In this case, the final shared state is a rank-four Werner state given by, 
\begin{equation}
    \rho^f_{\Lambda^4}=p|\Phi_{1}\rangle \langle \Phi_{1}|+\frac{1-p}{3}\sum_{i=2}^{4}|\Phi_{i}\rangle \langle \Phi_{i}|
\end{equation}
When $\frac{1}{2}< p < 1$, the maximal fidelity and the fidelity deviation of $\rho^f_{\Lambda^4}$ is given by, 
\begin{eqnarray}
F_{\rho^f_{\Lambda^4}}&=& \frac{2p+1}{3}>\frac{2}{3}, \nonumber \\
\Delta_{\rho^f_{\Lambda^4}}&=& 0. \nonumber
\end{eqnarray}
Hence, the final shared state is useful and universal for QT for a specific range of $p$. 

Next, consider another qubit channel $\tilde{\Lambda}^4$ associated with a rank-four Choi state having Kraus operators,
\begin{align}
&K^{\tilde{\Lambda}^4}_{0}=\sqrt{p_{0}}~\mathbf{I},~~K^{\tilde{\Lambda}^4}_{1}=\sqrt{p_{1}}~\sigma_{1}, \nonumber \\ &K^{\tilde{\Lambda}^4}_{2}=\sqrt{p_{2}}~\sigma_{2},~~K^{\tilde{\Lambda}^4}_{3}=\sqrt{p_{3}}~\sigma_{3}
\end{align}
with $\sum\limits_{i=0}^{3} p_{i}=1$ and $1> p_0 > p_1 > p_2 > p_3>0$. 
The final shared state is given by,
\begin{equation}
\rho^f_{\tilde{\Lambda}^4}=\sum_{i=0}^{3}p_{i}|\Phi_{i+1}\rangle \langle \Phi_{i+1}|.
\end{equation}
When $0 < p_0 \leq \frac{1}{2}$, the above final state is not useful for UQT \cite{FD}. On the other hand, for $\frac{1}{2} < p_0 < 1$,  the maximal fidelity and fidelity deviation can be written as 
\begin{eqnarray}
F_{\rho^f_{\tilde{\Lambda}^4}}&=&\frac{2p_{0}+1}{3} > \frac{2}{3}, \nonumber \\
\Delta_{\rho^f_{\tilde{\Lambda}^4}}&\neq & 0. \nonumber 
\end{eqnarray}
Hence, the final state is not useful for UQT in this case. The above two examples complete the proof.
\end{proof}

Next, we will characterize the set of unital as well as non-unital channels for which the final shared state will be useful for UQT.

\subsection{Alice sends one half of a Bell state via a unital channel}

Here we consider the scenario where Alice prepares the two-qubit Bell state $|\Phi_{1}\rangle$ and sends half of this state to Bob through a unital channel $\Lambda_u$. In this case, we present the following proposition.

\begin{proposition}
If Alice sends one half of a Bell state $|\Phi_{1} \rangle$ via any unital channel $\Lambda_u$, then the final shared state will be useful for UQT if and only if the channel is either unitary or a single parameter unital channel associated with rank-four Choi state (for a particular range of the channel parameter).
\end{proposition}

\begin{proof}
 If one half of the Bell state  $\Phi_{1}$ is sent via a unital channel $\Lambda_u$, then, up to local unitary rotations, the shared state after the channel interaction is given by \cite{QC,ruskai}, 
 \begin{eqnarray}
\rho^f_{u}&=& (\mathbf{I}\otimes \Lambda_{u})|\Phi_{1}\rangle \langle \Phi_{1}| \nonumber \\
&=& \sum_{i}(\mathbf{I}\otimes K_{i}^{\Lambda_{u}})|\Phi_{1}\rangle \langle \Phi_{1}|(\mathbf{I}\otimes K_{i}^{\Lambda_{u}})^{\dagger} \nonumber \\
&=&\sum_{i=0}^{3}p_{i}(\mathbf{I}\otimes \sigma_{i})|\Phi_{1}\rangle \langle \Phi_{1}|(\mathbf{I}\otimes \sigma_{i}) \nonumber \\
&=&  \sum_{i=0}^{3}p_{i}|\Phi_{i+1}\rangle \langle \Phi_{i+1}|, \label{BDform}
\end{eqnarray} 
with $0 \leq p_i \leq 1$ $\forall$ $i$, $\sum_{i=0}^{3} p_i =1$. Hence, the final state is Bell-diagonal up to local unitary rotations. Note that the above state $\rho^f_u$ is nothing but the Choi state associated with $\Lambda_u$. Now, rank-four Werner states (with a particular range of the state parameter) and rank-one maximally entangled pure states  are the only useful and universal states within Bell-diagonal class of states \cite{FD}. The final state (\ref{BDform}) will be rank-four Werner state if and only if $\Lambda_u$ satisfies $p_i = p_j = p_k =  \frac{1-p_l}{3}$ for any choice of $i \neq j \neq k \neq l \in \{0, 1, 2, 3\}$. Here, the rank-four Werner state will be useful and universal for QT when $\frac{1}{2} < p_l < 1$. That is, in this case, the unital channel is a one-parameter channel with rank-four Choi state having $p_i = p_j = p_k =  \frac{1-p_l}{3}$ for any choice of $i \neq j \neq k \neq l \in \{0, 1, 2, 3\}$ and $\frac{1}{2} < p_l < 1$. On the other hand, the final state (\ref{BDform}) will be a rank-one maximally entangled pure state if and only if $\Lambda_u$ satisfies $p_i = 1$ for any choice of $i \in \{0, 1, 2, 3\}$, which is nothing but a unitary channel.
\end{proof}

\subsection{Alice sends one half of a Bell state via non-unital channel}

Here we focus on non-unital qubit channels. When the initially prepared state is the Bell state $| \Phi_1\rangle$, we can state the following proposition,

\begin{proposition}
\textit{If Alice sends one half of a Bell pair $|\Phi_{1}\rangle$ through any non-unital channel $\Lambda_{nu}$, then the final shared state will be useful and universal for QT if and only if $\Lambda_{nu}$ belongs to a strict subset of non-unital qubit channels associated with rank-three and rank-four Choi states.}
\end{proposition}

\begin{proof}
Let Alice sends one half of a Bell state $|\Phi_{1}\rangle $ through a non-unital channel $\Lambda_{nu}$ to Bob. The final shared state $\rho^f_{nu}$ between Alice and Bob will be nothing but the Choi state $\rho_{\Lambda_{nu},\Phi_{1}}$ associated with the channel $\Lambda_{nu}$. The generic structure of $\rho^f_{nu}$ can be written as
\begin{align}
 \rho^f_{nu} &= (\mathbf{I}\otimes \Lambda_{nu})|\Phi_{1}\rangle \langle \Phi_{1}| \nonumber \\
 &=\frac{1}{4}\left[ \mathbf{I}_{4}+ \mathbf{I}\otimes \boldsymbol{s} \cdot \mathbf{\sigma}~+ \sum_{i=1}^{3} t_{ii}~\sigma_{i}\otimes \sigma_{i})\right], \label{nonchoi}
\end{align}
Note here that the marginal at Bob's end is given by $\text{Tr}_{1}(\rho^f_{nu})=\frac{1}{2}\left( \mathbf{I}+\sum_{i=1}^{3}s_{i}\sigma_{i} \right)$, where $\lbrace s_{1},s_{2},s_{3}\rbrace$ are local vector components with $|\boldsymbol{s}|=\sqrt{s_{1}^{2}+s_{2}^{2}+s_{3}^{2}}\neq 0$. 
\paragraph*{}
Now, $\rho^f_{nu}$ cannot be rank-one as channels associated with rank-one Choi states are unital \cite{ruskai}. Next, if $\rho^f_{nu}$ is a rank-two state, then Proposition \ref{prop2} tells that it will not be useful and universal for QT. Hence, $\rho^f_{nu}$ will be useful for UQT if it is a rank-three or rank-four state. That is, the final state will be useful and universal for QT if $\Lambda_{nu}$ is associated with rank-three or rank-four Choi states. The set of channels $\Lambda_{nu}$, for which the final state $\rho^{f}_{nu}$ is useful for UQT, is always a strict subset of non-unital qubit channels associated with rank-three and rank-four Choi states. The reason for the strictness is very simple. When  Alice sends one half of a Bell state to Bob through an arbitrary non-unital channel associated with rank-four/rank-three Choi state, the final state given by Eq.(\ref{nonchoi}) may not satisfy the condition $|t_{11}|=|t_{22}|=|t_{33}|$. Among them, the set of useful and universal states must satisfy $|t_{11}|=|t_{22}|=|t_{33}|=t>\frac{1}{3}$. 
\end{proof}
Next, let us evaluate the orthogonal Kraus operators of the most general non-unital channels for which the final state will be useful for UQT. 
\paragraph*{For non-unital channels with rank four Choi states:} At first, we consider $\Lambda_{nu}$ associated with rank-four Choi states. In this case, the final state given by Eq.(\ref{nonchoi}) will be useful for UQT if and only if $|t_{11}|=|t_{22}|=|t_{33}|=t>\frac{1}{3}$.  Without any loss of generality, let us assume the canonical representation of $\rho^{f}_{nu}$ for which $t_{11}=t_{22}=t_{33}=-t$ and $\frac{1}{3}< t \leq 1$. 

Since, the Choi states of $\Lambda_{nu}$ are rank-four states, the final states will also be rank four states. This will be satisfied when $\rho^{f}_{nu}$ satisfies the following inequality (see Appendix \ref{app1} for details),
\begin{eqnarray}
    |\boldsymbol{s}|&< & 1-t. \label{rank4cond} 
\end{eqnarray}
From the above condition and from the condition of non-unital channel: $|\boldsymbol{s}| > 0$, one can say that $t $ cannot be equal to one. Henceforth, we will consider $\frac{1}{3}< t < 1$. Therefore, in case of non-unital channels with rank-four Choi states, when the final state is universal, it cannot have maximal average fidelity equal to unity. 

Now, from the spectral decomposition of the final state $\rho^{f}_{nu}$ with the condition (\ref{rank4cond}), one can construct the set of four orthogonal Kraus operators $\lbrace K_{i}^{\Lambda^{4}_{nu}}\rbrace$. The explicit expressions of these Kraus operators are given by (see Appendix \ref{app1} for details), 
\begin{widetext}
\begin{align}
    K_{0}^{\Lambda^{4}_{nu}}&= x_0\left( \begin{array}{cc}
   \dfrac{is_{1}+s_{2}}{s_{3}-2t-\sqrt{|\boldsymbol{s}|^{2}+4t^{2}}}  & \dfrac{i(|\boldsymbol{s}|^{2}+s_{3}(s_{3}+2\sqrt{|\boldsymbol{s}|^{2}+4t^{2}}))}{|\boldsymbol{s}|^{2}-s_{3}^{2}+4s_{3}t} \\
   \\
   - i  & \dfrac{is_{1}-s_{2}}{-s_{3}+2t+\sqrt{|\boldsymbol{s}|^{2}+4t^{2}}}
\end{array} \right), 
\nonumber \\
   K_{1}^{\Lambda^{4}_{nu}} &= x_1 \left( \begin{array}{cc}
   \dfrac{-i(|\boldsymbol{s}|+s_{3})}{s_{1}+is_{2}}  & -i \\
   \\
   -i  & \dfrac{i(s_{3}-|\boldsymbol{s}|)}{s_{1}-is_{2}}
\end{array} \right), \
\nonumber \\
 K_{2}^{\Lambda^{4}_{nu}}&=x_2\left( \begin{array}{cc}
   \dfrac{is_{1}+s_{2}}{s_{3}-2t+\sqrt{|\boldsymbol{s}|^{2}+4t^{2}}}  & \dfrac{i(|\boldsymbol{s}|^{2}+s_{3}(s_{3}-2\sqrt{|\boldsymbol{s}|^{2}+4t^{2}}))}{|\boldsymbol{s}|^{2}-s_{3}^{2}+4s_{3}t} \\
   \\
  -i  & \dfrac{-is_{1}+s_{2}}{s_{3}-2t+\sqrt{|\boldsymbol{s}|^{2}+4t^{2}}}
\end{array} \right), 
\nonumber \\
 K_{3}^{\Lambda^{4}_{nu}} &=x_3\left( \begin{array}{cc}
  \dfrac{is_{1}+s_{2}}{|\boldsymbol{s}|+s_{3}}   & -i \\
  \\
   -i  & \dfrac{i(|\boldsymbol{s}|+s_{3})}{s_{1}-is_{2}}
\end{array} \right), \label{nonukraus4} 
\end{align}
\end{widetext}
where  $\frac{1}{3}<t<1$; $~~ 0<|\boldsymbol{s}| = \sqrt{s_1^2+s_2^2+s_3^2}<  1-t$ and 
\begin{small}
\begin{eqnarray}
x_{0}&=&  \frac{(2t-s_{3}+\sqrt{|\boldsymbol{s}|^{2}+4t^{2}})}{2\sqrt{2}}~\sqrt{\frac{(1+t+\sqrt{|\boldsymbol{s}|^{2}+4t^{2}})}{|\boldsymbol{s}|^{2}+2t(2t+\sqrt{|\boldsymbol{s}|^{2}+4t^{2}})}}, \nonumber \\
\nonumber \\
x_{1}&=&\frac{1}{2\sqrt{2}}~\frac{\sqrt{(|\boldsymbol{s}|^{2}-s_{3}^{2})(1+|\boldsymbol{s}|-t)}}{|\boldsymbol{s}|}, \nonumber \\
\nonumber \\
x_{2}&=&\frac{(s_{3}-2t+\sqrt{|\boldsymbol{s}|^{2}+4t^{2}})}{2\sqrt{2}}~\sqrt{\frac{(1+t-\sqrt{|\boldsymbol{s}|^{2}+4t^{2}})}{|\boldsymbol{s}|^{2}+2t(2t-\sqrt{|\boldsymbol{s}|^{2}+4t^{2}})}},\nonumber \\
\nonumber \\
x_{3}&=& \frac{1}{2\sqrt{2}}~\frac{\sqrt{(|\boldsymbol{s}|^{2}-s_{3}^{2})(1-|\boldsymbol{s}|-t)}}{|\boldsymbol{s}|}. \nonumber 
\label{factors}
\end{eqnarray}
\end{small}
The above matrices given by Eq.(\ref{nonukraus4}) representing the Kraus operators are expressed in the basis $\{|0\rangle, |1\rangle \}$.
Any non-unital channel with rank-four Choi state, for which the final state will be useful for UQT, belongs to the set of Channels with the above four orthogonal Kraus operators. It can be easily checked that the above Kraus operators satisfy the completeness property, $\sum_{i=0}^{3}(K_{i}^{\Lambda^{4}_{nu}})^{\dagger}~K_{i}^{\Lambda^{4}_{nu}}=\mathbf{I}$. The Choi states associated with set of the above Kraus operators are non-negative. Hence, these Kraus operators represent CPTP maps. Moreover, $\sum_{i=0}^{3}K_{i}^{\Lambda^{4}_{nu}}~(K_{i}^{\Lambda^{4}_{nu}})^{\dagger}\neq \mathbf{I}$ holds as long as $|\boldsymbol{s}| > 0$.
\paragraph*{}
\paragraph*{For non-unital channels with rank three Choi states:}
Next, we consider $\Lambda_{nu}$ associated with rank-three Choi states. In this case also, the final state given by Eq.(\ref{nonchoi}) will be useful for UQT if and only if $|t_{11}|=|t_{22}|=|t_{33}|=t>\frac{1}{3}$.  Without any loss of generality, we consider the canonical representation of $\rho^{f}_{nu}$ for which $t_{11}=t_{22}=t_{33}=-t$.    The final state $\rho^{f}_{nu}$ will be rank-three state if the following condition is satisfied (for details, see Appendix \ref{app1}), 
\begin{equation}
|\boldsymbol{s}|=\sqrt{s_{1}^{2}+s_{2}^{2}+s_{3}^{2}}=1-t, \label{rankthree}
\end{equation}
Since, $|\boldsymbol{s}| > 0$, we have $t < 1$, i.e., the final state cannot have maximal average fidelity equal to one. From the above condition, one can parametrize $\lbrace s_{i}\rbrace$ as follows, 
\begin{align}
    s_{1}&=(1-t)~\sin{\theta}\cos{\phi}, \nonumber \\
    s_{2}&=(1-t)~\sin{\theta}\sin{\phi}, \nonumber \\
    s_{3}&=(1-t)~\cos{\theta}, \nonumber
\end{align} 
where $\theta \in [0,\pi]$ and $\phi \in [0,2\pi]$. With these, one can construct a the following complete set of three orthogonal Kraus operators $\lbrace K_{i}^{\Lambda^{3}_{nu}}\rbrace $ (see Appendix \ref{app1}), 
\begin{widetext}
\begin{align}
 K_{0}^{\Lambda^{3}_{nu}}&= y_0\left( \begin{array}{cc}
   \dfrac{i(1-t)\sin{\theta}~e^{-i\phi}}{(1-t)\cos{\theta}-2t-\sqrt{(1-t)^{2}+4t^{2}}}  & \dfrac{i\left((1-t)^{2}(1+\cos^{2}{\theta})+2(1-t)\cos{\theta}\sqrt{(1-t)^{2}+4t^{2}}\right )}{(1-t)^{2}\sin^{2}{\theta}+4t(1-t)\cos{\theta}} \\
   \\
   -i  & \dfrac{i(1-t)\sin{\theta}e^{i\phi}}{-(1-t)\cos{\theta}+2t+\sqrt{(1-t)^{2}+4t^{2}}}
\end{array} \right), 
\nonumber \\
   K_{1}^{\Lambda^{3}_{nu}} &= y_1 \left( \begin{array}{cc}
   \dfrac{-i(1+\cos{\theta})}{\sin{\theta}e^{i\phi}}  & -i \\
   \\
   -i  & \dfrac{-i(1-\cos{\theta})}{\sin{\theta}e^{-i\phi}}
\end{array} \right), \
\nonumber \\
 K_{2}^{\Lambda^{3}_{nu}}&=y_2\left( \begin{array}{cc}
   \dfrac{i(1-t)\sin{\theta}e^{-i\phi}}{(1-t)\cos{\theta}-2t+\sqrt{(1-t)^{2}+4t^{2}}}  & \dfrac{i\left((1-t)^{2}(1+\cos^{2}{\theta})-2(1-t)\cos{\theta}\sqrt{(1-t)^{2}+4t^{2}}\right )}{(1-t)^{2}\sin^{2}{\theta}+4t(1-t)\cos{\theta}} \\
   \\
  -i  & \dfrac{-i(1-t)\sin{\theta}e^{i\phi}}{(1-t)\cos{\theta}-2t+\sqrt{(1-t)^{2}+4t^{2}}}
\end{array} \right), 
\label{nonukraus3gen}
\end{align}
 where $\frac{1}{3}<t<1$ and
 \begin{align}
    y_{0}&=  \frac{\left[2t-(1-t)\cos{\theta}+\sqrt{(1-t)^{2}+4t^{2}}\right]}{2\sqrt{2}}~\sqrt{\frac{(1+t+\sqrt{(1-t)^{2}+4t^{2}})}{(1-t)^{2}+2t(2t+\sqrt{(1-t)^{2}+4t^{2}})}}, \nonumber \\
   \nonumber \\
   y_{1}&=\frac{\sin{\theta}\sqrt{1-t}}{2}, \nonumber \\
   \nonumber \\
   y_{2}&=\frac{\left[(1-t)\cos{\theta}-2t+\sqrt{(1-t)^{2}+4t^{2}}\right]}{2\sqrt{2}}~\sqrt{\frac{(1+t-\sqrt{(1-t)^{2}+4t^{2}})}{(1-t)^{2}+2t(2t-\sqrt{(1-t)^{2}+4t^{2}})}}. \nonumber
 \end{align}
 \end{widetext}
Here also, the above matrices given by Eq.(\ref{nonukraus3gen}) representing the Kraus operators are expressed in the basis $\{|0\rangle, |1\rangle \}$. These Kraus operators $\lbrace K_{i}^{\Lambda^{3}_{nu}}\rbrace$ always satisfy the completeness property, $\sum_{i=0}^{2}(K_{i}^{\Lambda^{3}_{nu}})^{\dagger}~K_{i}^{\Lambda^{3}_{nu}}=~\mathbf{I}$ and also the  condition   $\sum_{i=0}^{2}K_{i}^{\Lambda^{3}_{nu}}~(K_{i}^{\Lambda^{3}_{nu}})^{\dagger}\neq~\mathbf{I}$ holds. The Choi states associated with the above set of Kraus operators are non-negative. Hence, the set of Kraus operators given by Eq.(\ref{nonukraus3gen}) represent CPTP maps associated with non-unital channels.

The above class of channels represents the most general non-unital channels with rank-three Choi states for which  the  final  state  will  be  useful for UQT.

Note that the set of Kraus operators associated with any non-unital channel, for which the final state will be useful for UQT, are always unitarily connected with the set of orthogonal Kraus operators given by Eq.(\ref{nonukraus4}) or Eq.(\ref{nonukraus3gen}).

Now, we will present specific examples of non-unital channels associated with rank three and rank four Choi states. Let us consider the non-unital quantum channel associated with rank-four Choi state having the following four orthogonal Kraus operators, 
\begin{align}
K_{0}^{\Lambda^{4}_{nu}}&=\sqrt{\dfrac{6+\sqrt{17}}{(17-\sqrt{17})}}\left( \begin{array}{cc}
    0 & 1 \\
    \dfrac{1-\sqrt{17}}{4} & 0
\end{array} \right), \nonumber \\
K_{1}^{\Lambda^{4}_{nu}} &=\left( \begin{array}{cc}
    \dfrac{\sqrt{3}}{2\sqrt{2}} & 0 \\
    0 & 0
\end{array} \right), \nonumber \\
K_{2}^{\Lambda^{4}_{nu}}&=\sqrt{\dfrac{6-\sqrt{17}}{(17+\sqrt{17})}}\left( \begin{array}{cc}
    0 & 1 \\
   \dfrac{1+\sqrt{17}}{4} & 0
\end{array} \right), \nonumber\\
K_{3}^{\Lambda^{4}_{nu}} &=\left( \begin{array}{cc}
    0 & 0
    \\
    0 & \dfrac{1}{2 \sqrt{2}}
\end{array} \right). \label{exampleXrankfour}
\end{align}

In this case, the maximal fidelity and fidelity deviation of the final state $\rho^f_{\Lambda^{4}_{nu}}$ are given by,   
\begin{align}
F_{\rho^f_{\Lambda^{4}_{nu}}}&=\frac{3}{4} > \dfrac{2}{3}, ~~~~~\Delta_{\rho^f_{\Lambda^{4}_{nu}}}=0, \nonumber 
\end{align}
which imply that the final state is useful for UQT.

Next we will show an example non-unital qubit channel with rank-four Choi state which does not preserve the universality criterion. For example, consider the generalized amplitude-damping channel $\Lambda^{\text{GADC}}$ \cite{GADC,SGAD} with the following Kraus operators,
\begin{align}
    K_{0}^{\text{GADC}}&= \sqrt{1-N}\left( \begin{array}{cc}
        1 & 0 \\
        0 & \sqrt{1-\gamma}
    \end{array} 
    \right), \nonumber \\
    K_{1}^{\text{GADC}} &=\sqrt{1-N}\left( \begin{array}{cc}
        0 & \sqrt{\gamma} \\
        0 & 0
    \end{array} 
    \right), \nonumber \\
     K_{2}^{\text{GADC}}&= \sqrt{N}\left( \begin{array}{cc}
        \sqrt{1-\gamma} & 0 \\
        0 & 1
    \end{array} 
    \right), \nonumber \\
    K_{3}^{\text{GADC}}&=\sqrt{N}\left( \begin{array}{cc}
        0 & 0 \\
        \sqrt{\gamma} & 0
    \end{array} 
    \right), \label{GADchannel}
\end{align}
where $\gamma,~N\in [0,1]$. This channel is non-unital when $\gamma (2N-1)\neq 0$, i.e., for $\gamma \neq 0$ and $N \neq \frac{1}{2}$. Note that the above Kraus operators are not orthogonal, but it can be checked that the Choi state associated with this Channel is rank-four.  In this case, the maximal fidelity and fidelity deviation of the final state $\rho^f_{\Lambda^{\text{GADC}}}$ are given by,   
\begin{align}
    F_{\rho^f_{\Lambda^{\text{GADC}}}}&=\frac{1}{2}+\frac{2\sqrt{1-\gamma}+(1-\gamma)}{6}, \nonumber \\
    & > \dfrac{2}{3} \quad \text{when} \quad \gamma < 2(\sqrt{2}-1), \nonumber\\ 
    \Delta_{\rho^f_{\Lambda^{\text{GADC}}}}&= \frac{\sqrt{1-\gamma}}{3\sqrt{5}}\left( 1-\sqrt{1-\gamma}\right), \nonumber
\end{align}
where $\Delta_{\rho_{GADC}}=0$ holds if and only if $\gamma =0$ or $\gamma =1$. But the final state becomes useless for QT when $\gamma =1$. On the other hand, for $\gamma =0$, the channel does not remain to be non-unital. Hence, the final state is not useful for UQT in case of non-unital GADC.

The above two examples illustrate that the set of non-unital qubit channels with rank-four Choi states, for which the final state remains to be useful for UQT, forms a strict subset of all non-unital channels associated with rank four Choi states. 

Next, let us consider the non-unital channel with rank-three Choi state having the  Kraus operators given by Eq.(\ref{nonkrausex}). As shown earlier, the final state in this case is useful  for UQT for a particular range of the channel parameter. 

Next, let us present another example of non-unital channel with rank-three Choi state having the following three orthogonal Kraus operators,
\begin{align}
K_{0}^{\Lambda^{3}_{nu}}&=\left( \begin{array}{cc}
    \dfrac{-3 \, i}{20} \sqrt{5 + 7 \sqrt{\dfrac{5}{17}}} \quad & \dfrac{i}{20} \sqrt{65 + 107 \sqrt{\dfrac{5}{17}}} \\
    \\
    \dfrac{-i}{20} \sqrt{65 + 107 \sqrt{\dfrac{5}{17}}} \quad & \dfrac{3 \, i}{20} \sqrt{5 + 7 \sqrt{\dfrac{5}{17}}}
\end{array} \right), \nonumber \\
K_{1}^{\Lambda^{3}_{nu}} &=\left( \begin{array}{cc}
    \dfrac{- 3 \, i}{2\sqrt{10}} \quad & \dfrac{- 3 \, i}{2\sqrt{10}} \\
    \\
    \dfrac{- 3 \, i}{2\sqrt{10}} \quad  & \dfrac{- 3 \, i}{2\sqrt{10}}
\end{array} \right), \nonumber \\
K_{2}^{\Lambda^{3}_{nu}} &=\left( \begin{array}{cc}
    \dfrac{3 \, i}{20} \sqrt{5 - 7 \sqrt{\dfrac{5}{17}}} \quad & \dfrac{i}{20} \sqrt{65 - 107 \sqrt{\dfrac{5}{17}}} \\
    \\
    \dfrac{-i}{20} \sqrt{65 - 107 \sqrt{\dfrac{5}{17}}} \quad & \dfrac{-3 \, i}{20} \sqrt{5 - 7 \sqrt{\dfrac{5}{17}}}
\end{array} \right). \label{exampleXrankthreenew}
\end{align}
In this case, the maximal fidelity and fidelity deviation of the final state $\rho^f_{\Lambda^{3}_{nu}}$ are given by,  
\begin{align}
F_{\rho^f_{\Lambda^{3}_{nu}}}&=\frac{11}{20} < \dfrac{2}{3}, ~~~~~\Delta_{\rho^f_{\Lambda^{3}_{nu}}}=0. \nonumber 
\end{align}
Hence, this final state is universal, but not useful for QT, i.e., this state is not useful for UQT.

\section{Two-qubit pure non-maximally entangled state as the initial state} \label{sec5}
Here we consider the  scenario where Alice prepares a non-maximally entangled two-qubit pure state given by, $|\Psi_{a}\rangle = \sqrt{a}|00\rangle +\sqrt{1-a} |11\rangle $ with $\frac{1}{2}<a<1$ and sends half of this state to Bob through a qubit channel $\Lambda$. The concurrence of the initial state $|\Psi_{a}\rangle$ is given by, $C(|\Psi_a\rangle) = 2\sqrt{a(1 - a)}$ with $0 < C(|\Psi_a\rangle) < 1$. In this case, the initially prepared state is useful and but not universal for QT \cite{FD}. Here, we want to analyse in details the characteristics of the final states in terms of maximal fidelity and fidelity deviation. We start with  the following result.

\begin{proposition}
When Alice sends one half of a non-maximally entangled two-qubit pure state to Bob via any qubit channel, then maximal fidelity of the final state will be less than or equal to that of the initial state.
\end{proposition}
\begin{proof}
Suppose Alice prepares a pure entangled state $|\Psi_{a}\rangle = \sqrt{a}|00\rangle +\sqrt{1-a} |11\rangle $, such that $\frac{1}{2}<a<1$. The maximal fidelity of this initial state is given by \cite{tele,FD},
\begin{align}
    F_{|\Psi_a\rangle}=\frac{2+C(|\Psi_a\rangle)}{3}. 
    \label{conini}
\end{align} 

When Alice sends one half of the state $|\Psi_{a}\rangle$ to Bob via any qubit channel, then the final state shared between Alice and Bob is denoted by $\rho^f$ and its concurrence is denoted by $C(\rho^f)$. Now, we have the following relation, $C(\rho^f) \leq C(|\Psi_a\rangle)$ as concurrence cannot be increased under local operations and classical communication \cite{conc}. It is well known that the maximal average fidelity $F_{\rho^{f}}$ is always upper bounded by \cite{uppbfidelity},
\begin{align}
    F_{\rho^{f}}&\leq ~\frac{2+C(\rho^{f})}{3} .
    \label{confin}
\end{align}
Hence, using relations (\ref{conini}), (\ref{confin}), we have
\begin{align}
    F_{\rho^{f}}&\leq \frac{2+C(|\Psi_a\rangle)}{3}  =F_{|\Psi_a\rangle}. \nonumber 
\end{align}
When the channel is an unitary channel (i.e., a particular unital channel), the above upper bound is saturated.
\end{proof}

Next, we focus on unital qubit channels.

\subsection{Alice sends one half of a non-maximally entangled two-qubit pure state through a unital channel}
At first, we present the following proposition which addresses the issue of usefulness and universality of the final state.
 \begin{proposition}
If Alice sends one half of any non-maximally entangled two-qubit pure state $|\Psi_{a}\rangle $ (with concurrence $C(|\Psi_a\rangle) \in (0,1)$) to Bob through any unital channel, then the final shared state will be useful for UQT for a strict subset of the unital channels if and only if $\frac{1}{2} < C(|\Psi_a\rangle)  < 1$. 
\end{proposition}
\begin{proof}
 If one half of the  state  $|\Psi_{a}\rangle $ is sent via an unital channel $\Lambda_u$, then, up to local unitary rotations, the shared state after the channel interaction is given by \cite{QC,ruskai}, 
 \begin{eqnarray}
\rho^f_{u}&=& (\mathbf{I}\otimes \Lambda_{u})|\Psi_{a}\rangle \langle \Psi_{a}| \nonumber \\
&=& \sum_{i}(\mathbf{I}\otimes K_{i}^{\Lambda_{u}})|\Psi_{a}\rangle \langle \Psi_{a}|(\mathbf{I}\otimes K_{i}^{\Lambda_{u}})^{\dagger} \nonumber \\
&=&\sum_{i=0}^{3}p_{i}(\mathbf{I}\otimes \sigma_{i})|\Psi_{a}\rangle \langle \Psi_{a}|(\mathbf{I}\otimes \sigma_{i}), \label{BDform}
\end{eqnarray} 
with $0 \leq p_i \leq 1$ $\forall$ $i$, $\sum_{i=0}^{3} p_i =1$. Up to local unitary transformations, without any loss of generality, we can assume that $p_0 \geq p_j$ $\forall$ $j \in \{1,2,3\}$. Now, it can be checked that the state $\rho^f_{u}$ belongs to the X-class of states \cite{x1,x2}. The expression of  concurrence for any X-class two-qubit state is known \cite{x2}. Using this, the concurrence of the state $\rho^f_{u}$ can be written as,
\begin{align}
C(\rho^f_{u})  &=  \max [0, (|p_1 - p_2| - p_0 - p_3) C(|\Psi_a\rangle), \nonumber \\
& \quad \quad \quad \quad \quad  (p_0 - p_1 - p_2 - p_3) C(|\Psi_a\rangle)].  
\end{align}
Now, we have $(|p_1 - p_2| - p_0 - p_3) \leq 0$ as long as $0 \leq p_i \leq 1$ $\forall$ $i \in \{0,1,2,3\}$ and $p_0 \geq p_j$ $\forall$ $j \in \{1,2,3\}$. Hence, $C(\rho^f_{u}) > 0$ if and only if $p_0 > p_1 + p_2 + p_3$. 

When $p_0 > p_1 + p_2 + p_3$ does not hold, the final state will not be entangled and, therefore, will not be useful for QT.  Henceforth, we will consider $p_0 > p_1 + p_2 + p_3$.

It can be verified that the correlation matrix of $\rho^f_{u}$ is  diagonal with the following eigenvalues,

\begin{eqnarray}
t_{11}&=& (p_{0}+p_{1}-p_{2}-p_{3}) \, C(|\Psi_a\rangle) = |t_{11}|, \nonumber \\
t_{22}&=& - (p_{0}-p_{1}+p_{2}-p_{3}) \, C(|\Psi_a\rangle) = - |t_{22}|, \nonumber \\
t_{33}&=&(p_{0}-p_{1}-p_{2}+p_{3}) = |t_{33}|. \label{corrunital}
\end{eqnarray}
Next, we use the relation, $p_3=1-p_0-p_1-p_2$. Now, the final state $\rho^f_{u}$ will be useful for UQT if and only if $|t_{11}|=|t_{22}|=|t_{33}|=t>\frac{1}{3}$. Using Eq.(\ref{corrunital}), it can be checked that the condition: $|t_{11}|=|t_{22}|=|t_{33}|=t$ is satisfied if and only if 
\begin{align}
p_1    = p_2 = \frac{1 + (1 - 2 p_0) \, C(|\Psi_a\rangle) }{4 + 2 C(|\Psi_a\rangle) }.
\end{align}
With these, we have
\begin{align}
|t_{11}|=|t_{22}|=|t_{33}|=t =   \frac{(4 p_0-1) C(|\Psi_a\rangle) }{2 + C(|\Psi_a\rangle) }
\end{align}
Now, $t >\frac{1}{3}$ will be satisfied if and only if
\begin{align}
   \frac{1}{4} < C(|\Psi_a\rangle)  &< 1 \quad \text{and} \nonumber \\
   \frac{1 + 2 C(|\Psi_a\rangle) }{6 C(|\Psi_a\rangle) } &< p_0 \leq 1 
\end{align}
With the above conditions, one can check that the following conditions hold: $0 \leq p_i \leq 1$ $\forall$ $i \in \{0,1,2,3\}$ and $p_0 > p_1 + p_2 + p_3$ if and only if 
\begin{align}
  \frac{1}{2} <& \, \, C(|\Psi_a\rangle)  < 1 \quad \text{and} \nonumber \\
    \frac{1 + 2 C(|\Psi_a\rangle) }{6 C(|\Psi_a\rangle) } <& p_0 \leq \, \frac{1}{2-C(|\Psi_a\rangle) }.   
\end{align}
The above conditions imply  that, for $\frac{1}{2} < C(|\Psi_a\rangle)  < 1$, the condition $|t_{11}|=|t_{22}|=|t_{33}|=t >\frac{1}{3}$ for the final state will not be satisfied for arbitrary choices of $p_0$, $p_1$, $p_2$, $p_3$. Only when $p_0$, $p_1$, $p_2$, $p_3$ satisfy some specific conditions, the final state will be useful for UQT. This completes the proof.
\end{proof}
The above proposition implies that for an arbitrary initial state $|\Psi_a \rangle$, one may not find any unital channel for which the final state is useful for UQT. But when the concurrence of the initial state is strictly greater than $\frac{1}{2}$, then one can always find a strict subset of the unital channels, for which the final state will be useful for UQT.
Hence, when one half of a non-maximally entangled two-qubit pure state with $\frac{1}{2} < C(|\Psi_a\rangle)  < 1$ is subjected to a particular unital channel, then there will be one disadvantage and one advantage. The disadvantage is that the maximal average fidelity of the final state will be less than or equal to that of the initial state. On the other hand, the advantage is that the final state will be useful for UQT though the initial state is not useful for UQT. Therefore, interaction of unital channel can reduce the fluctuation in fidelity which may have important information theoretic implications.

Next, we present two examples in support of the above proposition. Let us consider the unital channel $\Lambda_u^4$ associated with rank-four Choi state having the following Kraus operators,
\begin{align}
&K_0^{\Lambda_u^4} = \sqrt{\dfrac{2}{3}} \, \mathbf{I}, \nonumber \\   &K_1^{\Lambda_u^4} = \sqrt{\dfrac{3-p}{6(2+p)}} \, \sigma_1, \nonumber \\
&K_2^{\Lambda_u^4} = \sqrt{\dfrac{3-p}{6(2+p)}} \, \sigma_2, \nonumber \\ &K_3^{\Lambda_u^4} = \sqrt{\dfrac{2p-1}{3(2+p)}} \, \sigma_3, 
\end{align}
where $\frac{1}{2} \leq p <1$. Note that the above Kraus operators do not represent any CPTP map for $ p < \frac{1}{2}$ as the  associated Choi state becomes negative in this range.  Now, one half of the non-maximally entangled two-qubit pure state $|\Psi_{a}\rangle = \sqrt{a}|00\rangle +\sqrt{1-a} |11\rangle $ with $C(|\Psi_a\rangle) = 2\sqrt{a(1 - a)} = p$ is sent through the above channel.  The maximal average fidelity and fidelity deviation of the final shared state are given by,
\begin{align}
F_{\rho^f_{\Lambda_u^4}} &= \dfrac{3 + 4 \,  C(|\Psi_a\rangle)}{6+ 3 \, C(|\Psi_a\rangle)}\nonumber \\
&> \dfrac{2}{3} \quad \text{when} \quad \dfrac{1}{2} < C(|\Psi_a\rangle) < 1, \nonumber \\
\Delta_{\rho^f_{\Lambda_u^4}} &= 0 \quad \forall \, C(|\Psi_a\rangle) \in (0,1) . \nonumber
\end{align}
Hence, the final state in this case is useful for UQT when the concurrence of the initial state is strictly greater than $\frac{1}{2}$.

Next, consider the unital channel $\Lambda_u^2$ associated with rank-two Choi state having the following Kraus operators,
\begin{align}
&K_0^{\Lambda_u^2} = \sqrt{p} \, \mathbf{I}, \quad   &K_1^{\Lambda_u^2} = \sqrt{1-p} \, \sigma_3, \quad 0 < p <1.
\end{align}
One half of the state $|\Psi_{a}\rangle = \sqrt{a}|00\rangle +\sqrt{1-a} |11\rangle $ with $C(|\Psi_a\rangle) = 2\sqrt{a(1 - a)} \in (0,1)$ is sent through the above unital channel.  In this case, the maximal average fidelity and fidelity deviation of the final state are given by,
\begin{align}
F_{\rho^f_{\Lambda_u^2}} &= \left\lbrace \begin{array}{cc}
    \dfrac{2 +  \, C(|\Psi_a\rangle) (1- 2p)}{3}>\dfrac{2}{3} & \quad \text{when} \, \, 0< p < \dfrac{1}{2}, \\
    \\
    \dfrac{2}{3} & \quad \text{when} \, \, p= \dfrac{1}{2},\\
    \\
    \dfrac{2 +  \, C(|\Psi_a\rangle)  (2p-1)}{3}>\dfrac{2}{3} & \quad \text{when} \, \,  \dfrac{1}{2} < p <1,
\end{array}  
\right. \nonumber 
\end{align}
\begin{align}
\Delta_{\rho^f_{\Lambda_u^2}} &= \left\lbrace \begin{array}{cc}
    \dfrac{1- C(|\Psi_a\rangle)(1- 2p)}{3 \sqrt{5}} \neq 0 & \quad \text{when} \, \, 0< p < \dfrac{1}{2}, \\
    \\
    \dfrac{1}{3 \sqrt{5}} \neq 0 & \quad \text{when} \, \, p= \dfrac{1}{2},\\
    \\
    \dfrac{1- C(|\Psi_a\rangle)(2p-1)}{3 \sqrt{5}} \neq 0 & \quad \text{when} \, \, \dfrac{1}{2} < p <1,
\end{array}  
\right. \nonumber
\end{align}
Hence, the final state is not useful for UQT for any value of $a \in (\frac{1}{2}, 1)$.

\subsection{Alice sends one half of a non-maximally entangled two-qubit pure state through a non-unital channel}
Here, we consider that Alice prepares a non-maximally entangled two-qubit pure state $|\Psi_{a}\rangle = \sqrt{a}|00\rangle +\sqrt{1-a} |11\rangle $ with $\frac{1}{2}<a<1$ and sends one half of it to Bob through a non-unital qubit channel.

At first, let us take the non-unital channel $\Lambda_{nu}$ having the Kraus operators given by Eq.(\ref{nonkrausex}). It can be checked that the maximal average fidelity and fidelity deviation of the final shared state are given by,
\paragraph*{}
\paragraph*{}
\begin{widetext}
\begin{align}
F_{\rho^f_{\Lambda_{nu}}} &= \frac{1}{6}\left(3 + (1-p) \sqrt{1 - \left[C(|\Psi_a\rangle)\right]^2}  + p   + 2 p \, C(|\Psi_a\rangle) \right) \nonumber \\
&> \frac{2}{3} \quad \text{when} \quad 0 < C(|\Psi_a\rangle) < 1 \quad \text{and} \quad  \dfrac{2 + C(|\Psi_a\rangle) - 2 \sqrt{1 - \left[C(|\Psi_a\rangle)\right]^2}}{4 + 5 C(|\Psi_a\rangle)} < p <1, \nonumber \\
\Delta_{\rho^f_{\Lambda_{nu}}} &\neq 0 \quad \forall \,  C(|\Psi_a\rangle) \in (0,1) \quad \text{and} \quad \forall \, p \in (0,1). \nonumber 
\end{align}
Hence, the final state is not useful for UQT for all values of the concurrence of the initial state.

Next, let us take another example of non-unital channel $\widetilde{\Lambda}_{nu}$ having the following Kraus operators,
\begin{align}
K_0^{\widetilde{\Lambda}_{nu}} &=\dfrac{1}{\sqrt{2}} \sqrt{1-p_2-p_2 \dfrac{\sqrt{1-p_1}}{\sqrt{1+p_1}}} \left( \begin{array}{cc}
    0 \quad & 0 \\
    \\
    0 \quad & 1
\end{array} \right), \nonumber \\
K_1^{\widetilde{\Lambda}_{nu}} &=\dfrac{1}{\sqrt{2}} \sqrt{1-p_2+p_2 \dfrac{\sqrt{1-p_1}}{\sqrt{1+p_1}}}  \left( \begin{array}{cc}
    1 \quad & 0 \\
    \\
    0 \quad & 0
\end{array} \right), \nonumber \\
K_2^{\widetilde{\Lambda}_{nu}} &=\sqrt{\dfrac{1 + p_1 + p_2 + p_1 p_2 - p_2 \sqrt{1 + p_1} \sqrt{5 + 3 p_1} }{5 + 3 p_1 + \sqrt{1 - p_1} \sqrt{5 + 3 p_1}}}  \left( \begin{array}{cc}
    0 \quad & 1 \\
    \\
    \dfrac{\sqrt{1 - p_1} + \sqrt{5 + 3 p_1}}{2 \sqrt{1 + p_1}}
     \quad & 0
\end{array} \right), \nonumber \\
K_3^{\widetilde{\Lambda}_{nu}} &=\sqrt{\dfrac{1 + p_1 + p_2 + p_1 p_2 + p_2 \sqrt{1 + p_1} \sqrt{5 + 3 p_1} }{5 + 3 p_1 - \sqrt{1 - p_1} \sqrt{5 + 3 p_1}}}  \left( \begin{array}{cc}
    0 \quad & 1 \\
    \\
    \dfrac{\sqrt{1 - p_1} - \sqrt{5 + 3 p_1}}{2 \sqrt{1 + p_1}}
     \quad & 0
\end{array} \right), 
\end{align}
where $0 <p_1 < 1$ and $0 < p_2 < \dfrac{1 + p_1}{1 + p_1 + \sqrt{1 - p_1^2}}$. In this range, the above Kraus operators form a CPTP map corresponding to non-unital channel. Suppose, one half of the non-maximally entangled two-qubit pure state $|\Psi_{a}\rangle = \sqrt{a}|00\rangle +\sqrt{1-a} |11\rangle $ with $C(|\Psi_a\rangle) = 2\sqrt{a(1 - a)} = p_1$ is sent through the above channel. The maximal average fidelity and fidelity deviation of the final shared state are given by,
\begin{align}
F_{\rho^f_{\widetilde{\Lambda}_{nu}}} &= \dfrac{1 + p_2 \, C(|\Psi_a\rangle) }{2} \nonumber \\
&> \frac{2}{3} \quad \text{when} \quad \dfrac{\sqrt{17} -1}{6}  < C(|\Psi_a\rangle) < 1 \quad  \text{and} \quad  \dfrac{1}{3 C(|\Psi_a\rangle)} < p_2 < \dfrac{1 + C(|\Psi_a\rangle)}{1 + C(|\Psi_a\rangle) + \sqrt{1 - \left[C(|\Psi_a\rangle) \right]^2}}, \nonumber \\
\Delta_{\rho^f_{\widetilde{\Lambda}_{nu}}} &= 0 \quad   \forall \,  p_1 \in (0,1)  \quad \, \text{and} \quad \forall \, p_2 \in \left(0,\dfrac{1 + C(|\Psi_a\rangle)}{1 + C(|\Psi_a\rangle) + \sqrt{1 - \left[C(|\Psi_a\rangle) \right]^2}} \right). \nonumber 
\end{align}
\end{widetext}
Hence, the final state in this case is useful for UQT when the concurrence of the initial state is strictly greater than $\dfrac{\sqrt{17} -1}{6} \approx 0.52$.

Next, we consider another non-unital channel $\Lambda^{*}_{nu}$ having the following Kraus operators,
\begin{align}
K_0^{\Lambda^{*}_{nu}} &=\left( \begin{array}{cc}
    \sqrt{1-\gamma (p_{1})} \quad  & 0 \\
    \\
    0  \quad & 1
\end{array} \right), \nonumber \\
\nonumber \\
K_1^{\Lambda^{*}_{nu}} &=\left( \begin{array}{cc}
    0 \quad  & 0 \\
    \\
    \sqrt{\gamma(p_{1})} \quad & 0
    \end{array} \right),   
    \label{gadcnonmnew}
\end{align}
 where $\gamma (p_{1})= \dfrac{1 + \sqrt{1 - p_1^2} - \sqrt{ 3 p_1^2 - 2 + 2 \sqrt{1 - p_1^2}}}{2 + 2 \sqrt{1-p_1^2}}$ and  $0 < p_1 < 1$. Note that the above Kraus operators are obtained from the generalized amplitude-damping channel \cite{GADC} given by Eq.(\ref{GADchannel}) by putting $N=1$. Now, one half of the non-maximally entangled two-qubit pure state $|\Psi_{a}\rangle$ with concurrence $C(|\Psi_a\rangle) = 2\sqrt{a(1 - a)} = p_1$ is sent through the above channel. The maximal average fidelity and fidelity deviation of the final shared state in this case are given by,
\begin{widetext}
\begin{align}
F_{\rho^f_{\Lambda^{*}_{nu}}} &= \dfrac{3 - \sqrt{1 - \left[C(|\Psi_a\rangle) \right]^2} + \sqrt{ 3 \left[C(|\Psi_a\rangle) \right]^2 -2 + 2 \sqrt{1 - \left[C(|\Psi_a\rangle) \right]^2}}}{4} \nonumber \\
&> \frac{2}{3} \quad \text{when} \quad \sqrt{\dfrac{2}{3} \left(4 + \sqrt{3} \right) \left(1 - \frac{1}{6} \left(4 + \sqrt{3} \right) \right)}  < C(|\Psi_a\rangle)  < 1 \nonumber \\
\Delta_{\rho^f_{\Lambda^{*}_{nu}}} &= 0 \quad   \forall \quad   C(|\Psi_a\rangle) \in (0,1). \nonumber 
\end{align}
\end{widetext}
Hence, the final shared state is useful for UQT when the concurrence of the initial state is strictly greater than $\sqrt{\dfrac{2}{3} \left(4 + \sqrt{3} \right) \left(1 - \frac{1}{6} \left(4 + \sqrt{3} \right) \right)} \approx 0.41$.

We have analyzed with a number of analytical and numerical examples of non-unital channels. However, we have not found any non-unital qubit channel for which the final state is useful for UQT when the concurrence of the initial state is less than or equal to $0.41$. Hence, we can conjecture the following
\begin{conj}
If Alice sends one half of any pure entangled state $|\Psi_{a}\rangle $ to Bob through any non-unital channel, then the final shared state will be useful for UQT for a strict subset of the non-unital channels if and only if $C_{\text{Critical}} < C(|\Psi_a\rangle)  < 1$, where $0< C_{\text{Critical}} \leq 0.41$.
\end{conj}
We could not determine the precise value of $C_{\text{Critical}}$; further investigations are needed.

This conjecture demonstrates that when one-half of the state $|\Psi_{a}\rangle = \sqrt{a}|00\rangle +\sqrt{1-a} |11\rangle $ with concurrence $0.41 < C(|\Psi_a\rangle) = 2\sqrt{a(1 - a)} \leq \frac{1}{2}$ is sent through a qubit channel, then the final state will be useful for UQT if and only if the channel is non-unital. This result is in some sense counter-intuitive as it shows the advantage of  non-unital channels or dissipative interactions over unital or non-dissipative interactions in the context of UQT.

The examples presented here demonstrate that, similar to the case of unital channels, non-unital channels can decrease the fidelity deviation while acting upon one-half of a non-maximally entangled two-qubit pure state.

\section{Analysis with some physical noise models}\label{sec6}

Here we supplement our above developed results to quantum channels motivated by physical noise models. Quantum channels modeling these effects can be Markovian or non-Markovian, both unital as well as non-unital.  We consider that the Bell state $| \Phi_1\rangle$
is prepared and half of this entangled state is sent to a distant location through any of these channels. In this context, we want to find out the channels for which the final shared state will remain useful for UQT.

Among the examples of unital channels, we take the Random Telegraph Noise (RTN), the modified Ornstein-Uhlenbeck (OUN), Power law noises (PLN), depolarising, phase damping (PD) and non-Markovian dephasing (NMD). As representatives of non-unital channels, we study the amplitude damping and the Unruh channels. The Kraus representations of these channels along with the channel parameters, their impact on QT are presented in the tables of Appendix \ref{app2}.

Random Telegraph Noise (RTN) \cite{rice} is a local non-Markovian dephasing noise \cite{pradeep}.  The effect of RTN on the dynamics of quantum systems, specifically quantum correlations and control of open system dynamics has been studied in \cite{piilo,ali,pinto,control}. The autocorrelation function for the RTN, represented by the stochastic variable $\Gamma (t)$, is given by
\begin{equation}
\langle \Gamma (t) \Gamma (s) \rangle = a^2 e^{-|t-s|\gamma}, \label{rtn}
\end{equation}
where $a$ has the significance of the strength of the system-environment coupling and $\gamma$ is proportional to the fluctuation rate of the RTN. The RTN channel poses a well defined Markovian limit \cite{pradeep}.

The modified Ornstein-Uhlenbeck noise (OUN) is a well-known stationary Gaussian random process \cite{ornstein} which is in general a non-Markovian process.  This could model, for example, the spin of an electron interacting with a magnetic field subjected to random fluctuations \cite{yueberley}. Power  law  noise  (PLN),  also  called  $1/f^{\alpha}$ noise  is  a  non-Markovian  stationary  noise  process, where ${\alpha}$ is some real number.  A functional relationship exists between  the  spectral  density  and  the  frequency  of  the  noise \cite{falci}. Both the OUN and PLN have well-defined Markovian limits \cite{weakerthanCP}.

The non-Markovian depolarising channel is a generalization of the depolarizing channel to the case of colored noise \cite{wodco}.  Phase Damping channel models the phenomena  where  decoherence  occurs  without  dissipation (loss of energy) \cite{banghosh}.  Non-Markovian  dephasing channel is an extension  of  the  dephasing  channel  to  non-Markovian  class, identified with the breakdown in complete-positivity of the map \cite{NMDeph}.

A canonical non-unital channel which models both decoherence as well as dissipation is the amplitude damping channel (ADC). This can be modelled by a standard Lindblad type of master equation \cite{sbbook,SGAD} describing the evolution in the Markovian regime. When the decoherence rate is time-dependent, with a damped oscillatory behavior, then the amplitude damping noise is non-Markovian \cite{weakerthanCP}. With an appropriate change in the parameters, the Markovian limit of the channel can be easily derived. As another example of a non-unital channel, we take up the Unruh  channel \cite{bansrik}. To an observer undergoing acceleration,  the  Minkowski  vacuum  appears  as  a  warm  gas emitting  black-body  radiation  at  temperature, related to the acceleration, and called the Unruh temperature. This is known as the Unruh effect. 

From the analysis (see Appendix \ref{app2}) it can be seen that, among the channels considered, only for the unital depolarizing (Markovian or non-Markovian) noise channel, the final state is useful for UQT. The final state in this case  is nothing but a Werner class of state. In contrast, all the non-unital channels considered here create rank-two  final states. As discussed earlier, these rank-two final states cannot be useful for UQT.

\section{Conclusion}\label{sec7}

Apart from numerous applications, the idea of QT has refined various fundamental concepts of quantum information theory. Hence, from practical as well as from a foundational perspective, it is relevant to address the issue of suitably realizing quantum channels that can be used as the resource of QT. Motivated by this, in the present study we have considered a scenario where an observer, say, Alice prepares a two-qubit pure entangled state and sends one half of it to a distant observer, say, Bob through a quantum channel. The shared entangled state thus prepared is then used as a resource for QT. In this scenario, we have characterized the set of qubit channels in terms of the final shared state's performance in QT. 

In order to characterize the efficacy of the shared state in the context of QT, we have used two quantifiers- maximal average fidelity \cite{tele,can1,f1,f2} and fidelity deviation \cite{FD,Opt}. These two quantifiers together help to introduce the notion of ``useful states for Universal Quantum Teleportation (UQT)". A two-qubit state is called useful for UQT if and only if the maximal average fidelity is strictly greater than $\frac{2}{3}$ (i.e., the classical bound) and all input states are teleported with the same fidelity \cite{FD}. Let us now explain the significance of UQT from practical point of view. In real experiments, quantum teleportation can be realized as a single shot quantum gate operation with an input and an output \cite{gate.dist}. In such realistic situations, it is not expected that teleportations of all possible input states will be performed. Rather, teleportations of some particular input states are executed. In such cases, if the fidelity deviation of the two-qubit entangled channel is non-zero, then those particular  input states may be teleported with fidelity much less than the desired maximal average fidelity. However, two-qubit entangled channels with zero fidelity deviation can overcome such drawback- all input states will be teleported equally well. Most importantly, when the maximal average fidelity of a two-qubit entangled channel lies near the classical-quantum boundary (i.e., when the maximal average fidelity is greater than, but close to $\frac{2}{3}$) and fidelity deviation is non-zero, then some input states may be teleported with fidelity less than $\frac{2}{3}$. That is why the states with zero fidelity deviation should be preferred over states with nonzero fidelity deviation, especially near the quantum-classical boundary. Hence, with two-qubit entangled states having zero fidelity deviation, the gate operation will be universal or fluctuation free. Otherwise, one has to implement different gate operations for different choices of input states, which is problematic. 

Using the above notions,  we have shown that when half of a Bell state (which is useful for UQT) is sent through a unital or non-unital channel, then the final state is useful for UQT for a strict subset of channels. We further completely characterize these channels for which the final states are useful for UQT. Hence, these results indicate that a subset of unital as well as non-unital qubit channels can preserve the usefulness and universality of a maximally entangled state in the context of QT while acting upon one half of the state.

If one-half of a non-maximally entangled two-qubit pure state (which is \text{not} useful for UQT) is sent to Bob through a unital or non-unital channel, then we have demonstrated that the final state may become useful for UQT. It thus signifies that the action of a qubit channel on one half of a non-maximally entangled two-qubit pure state can turn it into useful for UQT. In case of unital channels, we have completely characterized such channels. On the other hand, we have shown the above in case of non-unital channels by presenting some specific examples. However, a complete characterization  of non-unital channels for which the final state is useful for UQT when the initial state is a non-maximally entangled two-qubit pure state merits further investigation. 

The set of channels which converts a pure non-maximally entangled two-qubit state into a state useful for UQT becomes more crucial when the input state is weakly entangled. The reason is that the maximal average fidelity of a weakly entangled pure two-qubit state lies near the quantum-classical boundary (i.e., the maximal average fidelity is slightly greater than $\frac{2}{3}$). On the other hand, weakly entangled pure two-qubit states possess a large amount of fidelity deviation. Hence, some input states in these cases are teleportated with fidelity less than $\frac{2}{3}$, though the  maximal average fidelity is greater than $\frac{2}{3}$. Thus, in these cases, it is much more desirable to choose quantum channels which completely eliminate the fidelity deviation and also keep the state useful for QT. 

The present study opens up several other directions for future research. Firstly, one should consider a different scenario where a two-qubit pure entangled state is initially prepared by an observer, say, Charlie and then one qubit is sent to Alice, another one is sent to Bob  (where Alice, Bob, Charlie- all are spatially separated from each other) through two different quantum channels. In this scenario, it is worth characterizing the set of qubit channels for which the final state will be useful for UQT. Apart from pure input states, the initially prepared state can be a mixed entangled state. It has been shown \cite{can2} that some specific non-unital channel interactions can increase the maximal average fidelity of some mixed two-qubit states. Therefore, it would be interesting to characterize quantum channels which not only increase maximal average fidelity but also eliminate fidelity deviation. Next, extending the present study to higher dimensional systems is another area for future research. Finally, we believe that our results will help in the experimental implementation of QT in realistic contexts. 

\section*{Acknowledgement}
AG and DD acknowledge fruitful discussions with Somshubhro Bandyopadhyay. AG acknowledges Bose Institute, Kolkata for financial support. AG acknowledges his visit at Indian Institute of Technology Jodhpur, where a part of this work was done. DD acknowledges Science and Engineering Research Board (SERB), Government of India for financial support through National Post Doctoral Fellowship (NPDF). SB acknowledges support from the Interdisciplinary Cyber Physical Systems (ICPS) program of the Department of Science and Technology (DST), India, Grant No.: DST/ICPS/QuEST/Theme-1/2019/6.

\newpage
\appendix
\begin{widetext}

\newpage
\section{Non-unital channels preserving the useful and universal condition while acting on one half of a Bell state} \label{app1}
Alice prepares a Bell state $|\Phi_{1}\rangle$ and sends one half to Bob via any non-unital channel $\Lambda_{nu}$ resulting the final state $\rho^{f}_{nu}$.  Here $\rho^f_{nu}$ is nothing but the Choi state of the channel $\Lambda_{nu}$. Now, $\rho^f_{nu}$ cannot be rank-one as channels associated with rank-one Choi states are unital \cite{ruskai}. Moreover, Proposition \ref{prop2} tells that $\rho^f_{nu}$ cannot be a rank-two state if it is useful for UQT. Henceforth, we will consider that $\rho^f_{nu}$ is either rank-three or rank-four state.

If $\rho^{f}_{nu}$ is useful for UQT, then the canonical density matrix of the state $\rho^{f}_{nu}$ is given by,
\begin{eqnarray}
    \rho^{f}_{nu}=\frac{1}{4}\left( \begin{array}{cccc}
1+s_{3}-t & s_{1}-i~s_{2} & 0 & 0 \\
s_{1}+i~s_{2} & 1-s_{3}+t & -2t & 0 \\ 
0 & -2t & 1+s_{3}+t & s_{1}-i~s_{2} \\
0 & 0 & s_{1}+i~s_{2} & 1-s_{3}-t
\end{array}
\right), \label{outputchoi}
\end{eqnarray}
where $|\boldsymbol{s}| = \sqrt{s_1^2+s_2^2+s_3^2} > 0$ (as the Channel is non-unital)  and we have taken $t_{11}=t_{22}=t_{33}=-t$ with $\frac{1}{3} < t \leq 1$. The spectral decomposition of $\rho^{f}_{nu}$ is given by, 
\begin{equation}
    \rho^{f}_{nu}=\sum_{i=0}^{3}q_{i}|\chi_{i}\rangle \langle \chi_{i}|,~~~~~~~~~~~~~\sum_{i}q_{i}=1,
\end{equation}
where $\lbrace q_{i}\rbrace$ is the set of eigenvalues of $\rho^{f}_{nu}$ given by,  
\begin{eqnarray}
q_{0}&=& q_{0}(s_{1},s_{2},s_{3},t)= \frac{1}{4}\left[ 1+t+\sqrt{|\boldsymbol{s}|^{2}+4t^{2}}\right], \nonumber \\
q_{1}&=&q_{1}(s_{1},s_{2},s_{3},t)=\frac{1}{4}\left[ 1+|\boldsymbol{s}|-t\right], \nonumber \\
q_{2}&=&q_{2}(s_{1},s_{2},s_{3},t)= \frac{1}{4}\left[ 1+t-\sqrt{|\boldsymbol{s}|^{2}+4t^{2}}\right], \nonumber \\
q_{3}&=&q_{3}(s_{1},s_{2},s_{3},t)=\frac{1}{4}\left[ 1-|\boldsymbol{s}|-t\right]. \label{eigenvalues}
\end{eqnarray}
And the set of eigenvectors $\lbrace |\chi_{i}\rangle \rbrace$ of $\rho^f_{nu}$ are given by,
\begin{eqnarray}
    |\chi_{0}\rangle &=& \frac{1}{\sqrt{n_{0}}}\left[  \frac{is_{1}+s_{2}}{s_{3}-2t-\sqrt{|\boldsymbol{s}|^{2}+4t^{2}}}~|00\rangle - ~i ~|01\rangle + \frac{i \left( |\boldsymbol{s}|^{2} + s_{3} \left( s_{3} + 2 \sqrt{|\boldsymbol{s}|^{2} + 4t^{2}} \right) \right)}{|\boldsymbol{s}|^{2}-s_{3}^{2}+4s_{3}t} ~ |10\rangle +\frac{is_{1}-s_{2}}{-s_{3}+2t+\sqrt{|\boldsymbol{s}|^{2}+4t^{2}}}~|11\rangle\right], \nonumber \\
    \nonumber \\
    |\chi_{1}\rangle &=&\frac{1}{\sqrt{n_{1}}}\left[  \frac{-i(|\boldsymbol{s}|+s_{3})}{s_{1}+is_{2}} |00\rangle  -i~ |01\rangle  -i~ |10\rangle + \frac{i(s_{3}-|\boldsymbol{s}|)}{s_{1}-is_{2}} |11\rangle\right], \nonumber \\
    \nonumber \\
    |\chi_{2}\rangle &=& \frac{1}{\sqrt{n_{2}}}\left[   \frac{is_{1}+s_{2}}{s_{3}-2t+\sqrt{|\boldsymbol{s}|^{2}+4t^{2}}}~|00\rangle -i ~ |01\rangle +\frac{i \left( |\boldsymbol{s}|^{2} + s_{3} \left( s_{3} - 2\sqrt{|\boldsymbol{s}|^{2}+4t^{2}} \right) \right)}{|\boldsymbol{s}|^{2}-s_{3}^{2}+4s_{3}t} ~|10\rangle  +\frac{-is_{1}+s_{2}}{s_{3}-2t+\sqrt{|\boldsymbol{s}|^{2}+4t^{2}}}~|11\rangle\right], \nonumber \\
    \nonumber \\
    |\chi_{3}\rangle &=&\frac{1}{\sqrt{n_{3}}}\left[  \frac{is_{1}+s_{2}}{|\boldsymbol{s}|+s_{3}} |00\rangle -i ~ |01\rangle -i ~ |10\rangle +\frac{i(s_{3} + |\boldsymbol{s}|)}{s_{1}-is_{2}} |11\rangle \right], \label{eigennu}
\end{eqnarray}
where $\lbrace n_{0},n_{1},n_{2},n_{3}\rbrace$ are the normalization factors of the eigenstates. The set of the normalized eigenstates $\lbrace |\chi_{i}\rangle \rbrace$ forms an orthonormal basis. Let us now define $\delta_{1}=q_{0}-q_{1}$, $\delta_{2}=q_{1}-q_{2}$ and $\delta_{3}=q_{2}-q_{3}$. Since the parameters $|\boldsymbol{s}|$ and $t$ are always non-negative numbers, one can verify that the conditions
\begin{align}
\delta_{1}&=2t+\sqrt{|\boldsymbol{s}|^{2}+4t^{2}}-|\boldsymbol{s}|>0 \label{a1inequality1} \\ \delta_{2}&=|\boldsymbol{s}|+\sqrt{|\boldsymbol{s}|^{2}+4t^{2}}-2t> 0, \label{a1inequality2} \\
\delta_{3}&=|\boldsymbol{s}|+2t-\sqrt{|\boldsymbol{s}|^{2}+4t^{2}}>0 ,~~~~~~~\label{a1inequality3}
\end{align}
always hold. The proof is straightforward. For any two positive real numbers $x$ and $y$ with $0<x<1$ and $0<y<1$, one can always verify that the set of inequalities 
\[
x+\sqrt{x^{2}+y^{2}}>y,~~~~y+\sqrt{x^{2}+y^{2}}>x,~~\text{and}~~x+y>\sqrt{x^{2}+y^{2}}
\]
always hold. Now substituting $x=2t$ and $y=|\boldsymbol{s}|$, we obtain the three conditions (\ref{a1inequality1}, \ref{a1inequality2}, \ref{a1inequality3}). Hence, the ordering of the eigenvalues are given by, $q_{0}> q_{1}> q_{2}> q_{3}$. Therefore, the positivity of $\rho^{f}_{nu}$ implies that
\[
q_{3}\geq 0 \implies|\boldsymbol{s}|\leq 1-t. 
\]

Now, if the final state is rank four, then we have $q_{3}>0$ implying $|\boldsymbol{s}|<1-t$ whereas for rank-three final state we have $q_{3}=0$ implying $|\boldsymbol{s}|=1-t$.

\paragraph*{Estimation of the orthogonal Kraus operators:}
Next, we will evaluate the orthogonal Kraus operators associated with the non-unital Channels preserving the useful and universal condition while acting on one half of a Bell state $|\Phi_1\rangle$ following the approach mentioned in \cite{can1,QC}. 

After the action of an arbitrary non-unital channel, any eigenvector of the final state $\rho^{f}_{nu}$ can be written as 
\[
|\chi_{i}\rangle =\sum_{m,n=0}^{1}a^{(i)}_{mn}|m\rangle ~|n\rangle,~~~~~~~~~~~~~\text{with} ~~~~~~~~~~~~~ \sum_{m,n=0}^{1}|a^{(i)}_{mn}|^{2}=1.
\]

Now, one can define a $2\times 2$ complex matrix $A_{i}$ given by,
\begin{align}
A_{i}=\sqrt{2}\left( 
\begin{array}{cc}
    a_{00}^{(i)} & a_{01}^{(i)} \\
    \\
    a_{10}^{(i)} & a_{11}^{(i)}
\end{array}
\right),~~~~~~~~~~~~~~i=0,1,2,3.
\end{align}
It can be easily checked that 
\begin{align}
    |\chi_{i}\rangle =(A_{i}\otimes \mathbf{I})|\Phi_{1}\rangle = (\mathbf{I}\otimes A_{i}^{T})|\Phi_{1}\rangle, \label{krauscon}
\end{align}
where $A_{i}^{T}$ is the transposition of $A_i$. Hence, from Eq.(\ref{krauscon}), we can write
\begin{align}
    \rho^{f}_{nu} &=\sum_{i=0}^{3}q_{i}|\chi_{i}\rangle \langle \chi_{i}| \nonumber \\
    &=\sum_{i=0}^{3}(\mathbf{I}\otimes \sqrt{q_i} A_{i}^{T})|\Phi_{1}\rangle \langle \Phi_{1}|(\mathbf{I}\otimes \sqrt{q_i} A_{i}^{T})^{\dagger},
\end{align}
where $q_i$ are real positive numbers denoting the eigenvalues of $\rho^{f}_{nu}$ with $\sum_{i} q_i =1$; $\lbrace |\chi_{i}\rangle \rbrace$ are the orthonormal eigenvectors of $\rho^{f}_{nu}$. Here $\{ \sqrt{q_i} A_{i}^{T} \}$ denotes the Kraus operators of the channel $\Lambda_{nu}$. Note that the Kraus operators $\{ \sqrt{q_i} A_{i}^{T} \}$  correspond to completely positive trace preserving maps \cite{can1,QC}. 

Next, one can easily varify the following condition,
\begin{align}
    \text{Tr}\left[ \left(\sqrt{q_i} A_{i}^{T}\right)^{\dagger} \left( \sqrt{q_j} A_{j}^{T} \right) \right] = q_i \delta_{ij}, 
\end{align}
where $\delta_{ij}$ is the Kronecker delta function. In the above, we have used the fact that the eigenvectors $\lbrace |\chi_{i}\rangle \rbrace$ are orthonormal. Hence,  the Kraus operators $\{ \sqrt{q_i} A_{i}^{T} \}$ are orthogonal Kraus operators.

Now, using the expressions of eigenvalues and eigenvectors given by (\ref{eigenvalues}) and (\ref{eigennu}), and using the condition $|\boldsymbol{s}| < 1-t$ for non-unital channels with rank-four Choi states and the condition $|\boldsymbol{s}|= 1-t$ for non-unital channels with rank-three Choi states, one can evaluate the most general form of orthogonal Kraus operators $\{ \sqrt{q_i} A_{i}^{T} \}$ of non-unital channels that preserve the usefulness and universality condition.

\newpage

\section{Details of physical noise models and their impact on QT} \label{app2}
The details of some physically motivated noise models, as discussed in Sec. \ref{sec6}, are summarized in the below table. 

\begin{small}
\begin{center}
 \begin{tabular}{|c |c |c |c |} 
 \hline
 Type of noise & Nature & Kraus operators & Channel parameters  \\ [0.5ex] 
 \hline\hline
 & & &  \\
 Depolarizing (Markovian) & Unital & $M_{0}=\sqrt{1-p}$ $\mathbf{I}$,$M_{i}=\sqrt{\dfrac{p}{3}}\sigma_{i}$, $i=1,2,3 $ & $p\in [0,1]$   \\ 
  & & &  \\
 \hline
  & & &  \\
 Dephasing (Markovian) & Unital & $M_{0}=\sqrt{1-p}$ $\mathbf{I}$, $M_{1}=\sqrt{p}\sigma_{3}$ & $p\in [0,1]$  \\
  & & &  \\
 \hline
  & & &  \\
 ADC (Markovian) & Non-unital &$M_{0}= \left( \begin{array}{cc}
 1 & 0 \\
 0 & \sqrt{1-p(t)}
\end{array} \right)$ , $M_{1}= \left( \begin{array}{cc}
 0 & \sqrt{p(t)} \\
 0 & 0
\end{array} \right)$   & $p(t)=1- \exp \left[ -\gamma t \right]$  \\
 & & &  \\
 \hline
  & & & \\
 PLN (Markovian) & Unital & $M_{0}=\sqrt{1-p(t)}$ $\mathbf{I}$, $M_{1}=\sqrt{p(t)}\sigma_{3}$ & $p(t)=\exp \left[ -Gt \right]$ \\
  & & &  \\
 \hline
  & & &  \\
 OUN (Markovian) & Unital & $M_{0}=\sqrt{1-p(t)}$ $\mathbf{I}$, $M_{1}=\sqrt{p(t)}\sigma_{3}$ & $p(t)=\exp \left[-\dfrac{Gt}{2} \right]$ \\ 
  & & &  \\
 \hline
  & & &  \\
 Unruh (Markovian) & Non-unital & $M_{0}= \left( \begin{array}{cc}
 \cos r & 0 \\
 0 & 1
\end{array} \right)$ , $M_{1}= \left( \begin{array}{cc}
 0 & 0 \\
 \sin r & 0
\end{array} \right)$ & $r\in (0,\frac{\pi}{4}]$ \\
 & & &  \\
 \hline
  & & &  \\
 Depolarizing (Non-Markovian) & Unital & $M_{0}=\sqrt{[(1-3\alpha p)(1-p)]}~\mathbf{I}$,  & $0<\alpha \leq 1$,\\ 
  & & &  \\
  & & $M_{i}=\sqrt{[1+3\alpha(1-p)]\frac{p}{3}}~\sigma_{i}$, \quad $i=1,2,3$ & $0\leq p \leq \frac{1}{2}$\\
  & & & \\
 \hline
  & & &  \\
 Dephasing (Non-Markovian) & Unital & $M_{0}=\sqrt{(1-\alpha~p)(1-p)}\mathbf{I} $, & $0<\alpha \leq 1$, \\
 & & & \\
  & & $M_{3}=\sqrt{p\left[ 1+\alpha~(1-p)\right]}\sigma_{3}$ & $0\leq p \leq \frac{1}{2}$ \\
  & & & \\
 \hline
  & & &  \\
 ADC (Non-Markovian) & Non-unital & $M_{0}= \left( \begin{array}{cc}
 1 & 0 \\
 0 & \sqrt{1-p(t)}
\end{array} \right)$ , $M_{1}= \left( \begin{array}{cc}
 0 & \sqrt{p(t)} \\
 0 & 0
\end{array} \right)$ & $p(t)=1- \exp \left[\dfrac{-2R\gamma}{\omega_{0}\coth \left( \frac{g\omega_{0}t}{2} \right) +1} \right]$ \\
 & & &  \\
 \hline
  & & &  \\
 PLN (Non-Markovian) & Unital & $M_{0}=\sqrt{1-p(t)}$ $\mathbf{I}$, $M_{1}=\sqrt{p(t)}\sigma_{3}$ & $p(t)=\exp \left[\dfrac{Gt(gt+2)}{2(gt+1)^{2}} \right]$ \\
  & & &  \\
 \hline
  & & &  \\
 OUN (Non-Markovian) & Unital & $M_{0}=\sqrt{1-p(t)}$ $\mathbf{I}$, $M_{1}=\sqrt{p(t)}\sigma_{3}$ & $p(t)=\exp \left[\dfrac{-G(g^{-1}(e^{-gt}-1)+t)}{2} \right]$  \\  
  & & &  \\
 \hline
\end{tabular}
\end{center}
\end{small}

\begin{small}
\begin{center}
 \begin{tabular}{|c |c |c |c |} 
 \hline
 Type of noise & Nature & Kraus operators & Channel parameters  \\ [0.5ex] 
 \hline\hline
 & & &  \\
 RTN (Non-Markovian) & Unital & $M_{0}=\sqrt{1-p(t)}$ $\mathbf{I}$, $M_{1}=\sqrt{p(t)}\sigma_{3}$  & $p(t)= \exp \left[ -gt \right]\left( \cos(g\omega t)+\dfrac{\sin(g\omega t)}{\omega}\right) $ \\
  & & &  \\
 \hline
\end{tabular}
\end{center}
\end{small}



Next, we will present the maximal fidelity and fidelity deviation of the final states when one half of the state $|\Phi_1\rangle$ is subjected to the aforementioned channels. Note that if a state $\rho$ has $\text{det}(T_{\rho}) \geq 0$ (with $T_{\rho}$ being the correlation matrix of $\rho$), then that state is not useful for QT \cite{hs2,can2}. That is why, in the following, we will present the the expressions for maximal fidelity and fidelity deviation of the final states only for those ranges of channel parameters where the final states satisfy $\text{det}(T_{\rho}) < 0$. In the below tabel, `M' stands for Markovian channels and `NM' stands for non-Markovian channels.
\begin{center}
\begin{tabular}{|c |c |c |c |}
\hline
 & Maximal fidelity & Fidelity deviation & Whether the final \\[0.5 ex]
Type of noise & of the & of the &  state is useful \\[0.5 ex]
& final state & final state & for UQT for some\\[0.5 ex]
 & & & specific ranges of the \\[0.5ex]
  & & & channel parameters \\[0.5ex]
\hline \hline
Depolarizing & $1-\dfrac{2p}{3}$ & 0 & Yes \\[2.5 ex]
 (M) & $>\dfrac{2}{3}$ when $0\leq p<\dfrac{1}{2}$ &  &\\[2.5 ex]
\hline
Dephasing & $ \dfrac{2p+1}{3}>\dfrac{2}{3}$  when $\dfrac{1}{2}< p < 1$; &  $\dfrac{2(1-p)}{3\sqrt{5}} \neq 0$  when  $\dfrac{1}{2}< p < 1$; & \\[2.5 ex]
(M)    & $\dfrac{2}{3}    \hspace{0.75cm}$ when $p= \dfrac{1}{2}$; & $\dfrac{1}{3\sqrt{5}} \neq 0  \hspace{0.1cm}$ when  $p = \dfrac{1}{2}$; & No\\[2.5 ex]
 &    $1-\dfrac{2p}{3}>\dfrac{2}{3}$   when $0< p < \dfrac{1}{2}$. & $\dfrac{2p}{3\sqrt{5}} \neq 0$ when $0< p < \dfrac{1}{2}$. & \\[2.5 ex]
\hline
ADC (M) & $\dfrac{1}{2}+\dfrac{\exp\left[-\gamma t\right] + 2 \exp\left[-\frac{\gamma t}{2}\right]}{6}$ & $\dfrac{\exp \left[-\frac{\gamma t}{2} \right] - \exp \left[-\gamma t\right]}{3\sqrt{5}}$ &   \\ [4ex]
& $>\dfrac{2}{3}$ \, if and only if  & $\neq 0$ \, when   & No \\[2.5 ex]
&  \, $0 < t < \dfrac{-ln(3-2\sqrt{2})}{\gamma}$ & $0 < t < \dfrac{-ln(3-2\sqrt{2})}{\gamma}$ & \\[2.5 ex]
\hline
PLN (M) & $\dfrac{2e^{-Gt}+1}{3}>\dfrac{2}{3}$   when $0\leq t < \dfrac{ln~2}{G}$; & $\dfrac{2(1-e^{-Gt})}{3\sqrt{5}}\neq 0$   when $0\leq t < \dfrac{ln~2}{G}$; &  \\[2.5 ex]
 &   $\dfrac{2}{3}  \hspace{1.7cm}$ when $t= \dfrac{ln~2}{G}$; & $\dfrac{1}{3\sqrt{5}}\neq 0  \hspace{1cm}$ when  $t= \dfrac{ln~2}{G}$; & No\\[2.5 ex]
  &  $1-\dfrac{2e^{-Gt}}{3}>\dfrac{2}{3}$ when $\dfrac{ln~2}{G} < t < \infty$. & $\dfrac{2e^{-Gt}}{3\sqrt{5}}\neq 0$ when $\dfrac{ln~2}{G} < t < \infty$. & \\[2.5 ex]
\hline
\end{tabular}
\end{center}

\begin{center}
\begin{tabular}{|c |c |c |c |}
\hline
 & Maximal fidelity & Fidelity deviation & Whether the \\[0.5 ex]
Type of noise & of the & of the &   final state \\[0.5 ex]
& final state & final state &  is useful \\[0.5 ex]
 & & & for UQT  \\[0.5ex]
  & & & for some \\[0.5ex]
   & & & specific ranges  \\[0.5ex]
    & & & the channel \\[0.5ex]
     & & & parameters \\[0.5ex]
\hline \hline
OUN (M) &  $\dfrac{2e^{-\frac{Gt}{2}}+1}{3}>\dfrac{2}{3}$  & $\dfrac{2(1-e^{-\frac{Gt}{2}})}{3\sqrt{5}}\neq 0 $ & \\[2.5 ex]
 &  when $0\leq t < \dfrac{ln~2}{G}$; &  when $0\leq t < \dfrac{2ln~2}{G}$; & \\[2.5 ex]
  &  $\dfrac{2}{3} \hspace{0.5cm}$ when $t= \dfrac{2ln~2}{G}$; & $\dfrac{1}{3\sqrt{5}}\neq 0 \hspace{0.2cm}$ when $t= \dfrac{2ln~2}{G}$; & No \\[2.5 ex]
 &   $1-\dfrac{2e^{-\frac{Gt}{2}}}{3}>\dfrac{2}{3}$ & $\dfrac{2e^{-\frac{Gt}{2}}}{3\sqrt{5}}\neq 0 $  & \\[2.5 ex]
  &   when $\dfrac{2ln~2}{G} < t < \infty$. &  when $\dfrac{2ln~2}{G} < t < \infty$. & \\[2.5 ex]
\hline
Unruh(M) & $ \dfrac{1}{2}+\dfrac{\cos^{2} r +2 \cos r}{6}$ & $\dfrac{\cos r - \cos^{2} r}{3\sqrt{5}}$&  \\[2.5ex]
 & $ >\dfrac{2}{3}$ \quad if and only if \, \, $0<r\leq \dfrac{\pi}{4}$ & $\neq 0$ \quad when \, $0<r\leq \dfrac{\pi}{4}$& No \\[2.5ex]
\hline
Depolarizing & $1-\dfrac{2p}{3}\left[ 1+3\alpha ~(1-p)\right]$ &  &  \\[2.5ex]
 (NM) & $> \dfrac{2}{3}$ \quad when & 0 & Yes \\[2.5ex]
 & $0\leq p<\dfrac{1}{6\alpha}\left[1+3\alpha-\sqrt{1+9\alpha ^{2}}\right]$, & & \\[2.5ex]
 & $0<\alpha \leq 1$ & & \\[2.5ex]
\hline
Dephasing & $1-\dfrac{2p\left[ 1+\alpha~(1-p)\right]}{3}>\dfrac{2}{3}$  & $\dfrac{2\left[1-(1-p)(1-\alpha~p)\right]}{3\sqrt{5}}\neq 0$  &  \\[2.5ex]
(NM) &  when &  when &  \\[2.5ex]
  & $0\leq p<~ \dfrac{1}{2\alpha}\left[1+\alpha-\sqrt{1+\alpha ^{2}}\right]$; & $ 0\leq p<~ \dfrac{1}{2\alpha}\left[1+\alpha-\sqrt{1+\alpha ^{2}}\right]$; & \\[2.5ex]
  & $\dfrac{2}{3}$ when $p= \dfrac{1}{2\alpha}\left[1+\alpha-\sqrt{1+\alpha ^{2}}\right]$; & $\dfrac{\sqrt{5}}{3}$ when $p= \dfrac{1}{2\alpha}\left[1+\alpha-\sqrt{1+\alpha ^{2}}\right]$; & \\[2.5ex]
  &  $\dfrac{1}{3}\left[ 1+2p(1+\alpha -\alpha~p)\right]>\dfrac{2}{3}$ & $\dfrac{2\left[1-(p)(1+\alpha-\alpha~p)\right]}{3\sqrt{5}}\neq 0$  & \\[2.5ex]
  &  when &  when & No \\[2.5ex]
  & $\dfrac{1}{2\alpha}\left[1+\alpha-\sqrt{1+\alpha ^{2}}\right]<p<~\dfrac{1}{2}$;  & $\dfrac{1}{2\alpha}\left[1+\alpha-\sqrt{1+\alpha ^{2}}\right]<p<~\dfrac{1}{2}$;  & \\[2.5ex]
  & $0<\alpha \leq 1$  & $0<\alpha \leq 1$  & \\[2.5ex]
 \hline

\end{tabular}
\end{center}

\begin{center}
\begin{tabular}{|c |c |c |c |}
\hline
 & Maximal fidelity & Fidelity deviation & Whether the \\[0.5 ex]
Type of noise & of the & of the &   final state \\[0.5 ex]
& final state & final state &  is useful \\[0.5 ex]
 & & & for UQT  \\[0.5ex]
  & & & for some \\[0.5ex]
   & & & specific ranges  \\[0.5ex]
    & & & the channel \\[0.5ex]
     & & & parameters \\[0.5ex]
\hline 
 \hline
ADC (NM) & $\dfrac{1}{2}+\dfrac{2\sqrt{1-p(t)}+1-p(t)}{6}$  & $\dfrac{\sqrt{1-p(t)}\left(1-\sqrt{1-p(t)}\right)}{3\sqrt{5}}\neq 0$ &  \\
[2.5ex]
 & $>\dfrac{2}{3}$ \quad if and only if &  when & No \\
[2.5ex]
 & $0<p(t)\leq 1-exp\left[ \dfrac{-2R\gamma}{\omega_{0}+1}\right]$ & $0<p(t)\leq 1-exp\left[ \dfrac{-2R\gamma}{\omega_{0}+1}\right]$ & \\[2.5ex]
 & $<2(\sqrt{2}-1)$ & $<2(\sqrt{2}-1)$ & \\[2.5ex]
 & where & where & \\[2.5ex]
  & $p(t)=1- \exp \left[\dfrac{-2R\gamma}{\omega_{0}\coth \left( \frac{g\omega_{0}t}{2} \right) +1} \right]$ & $p(t)=1- \exp \left[\dfrac{-2R\gamma}{\omega_{0}\coth \left( \frac{g\omega_{0}t}{2} \right) +1} \right]$ & \\[2.5ex]
\hline
PLN (NM) & $ \dfrac{2p(t)+1}{3}>\dfrac{2}{3}$  when $\dfrac{1}{2}< p(t) < 1$, &  $\dfrac{2(1-p(t))}{3\sqrt{5}} \neq 0$  when  $\dfrac{1}{2}< p(t) < 1$, & \\[2.5 ex]
   & $\dfrac{2}{3}    \hspace{0.75cm}$ when $p(t)= \dfrac{1}{2}$, & $\dfrac{1}{3\sqrt{5}} \neq 0  \hspace{0.1cm}$ when  $p(t) = \dfrac{1}{2}$, & No\\[2.5 ex]
 &    $1-\dfrac{2p(t)}{3}>\dfrac{2}{3}$   when $0< p(t) < \dfrac{1}{2}$, & $\dfrac{2p(t)}{3\sqrt{5}} \neq 0$ when $0< p(t) < \dfrac{1}{2}$ & \\[2.5 ex]
 & where  & where &\\[2.5 ex]
 & $p(t)= \exp \left[\dfrac{Gt(gt+2)}{2(gt+1)^{2}} \right]$ &  $p(t)= \exp \left[\dfrac{Gt(gt+2)}{2(gt+1)^{2}} \right]$ & \\[2.5 ex]
 \hline
OUN (NM) & $ \dfrac{2p(t)+1}{3}>\dfrac{2}{3}$  when $\dfrac{1}{2}< p(t) < 1$, &  $\dfrac{2(1-p(t))}{3\sqrt{5}} \neq 0$  when  $\dfrac{1}{2}< p(t) < 1$, & \\[2.5 ex]
   & $\dfrac{2}{3}    \hspace{0.75cm}$ when $p(t)= \dfrac{1}{2}$, & $\dfrac{1}{3\sqrt{5}} \neq 0  \hspace{0.1cm}$ when  $p(t) = \dfrac{1}{2}$, & No\\[2.5 ex]
 &    $1-\dfrac{2p(t)}{3}>\dfrac{2}{3}$   when $0< p(t) < \dfrac{1}{2}$, & $\dfrac{2p(t)}{3\sqrt{5}} \neq 0$ when $0< p(t) < \dfrac{1}{2}$ & \\[2.5 ex]
 & where  & where & \\[2.5 ex]
 & $p(t)= \exp \left[\dfrac{-G(g^{-1}(e^{-gt}-1)+t)}{2} \right]$ &  $p(t)= \exp \left[\dfrac{-G(g^{-1}(e^{-gt}-1)+t)}{2} \right]$ &\\[2.5 ex]
\hline
 RTN (NM) & $ \dfrac{2p(t)+1}{3}>\dfrac{2}{3}$  when $\dfrac{1}{2}< p(t) < 1$, &  $\dfrac{2(1-p(t))}{3\sqrt{5}} \neq 0$  when  $\dfrac{1}{2}< p(t) < 1$, & \\[2.5 ex]
   & $\dfrac{2}{3}    \hspace{0.75cm}$ when $p(t)= \dfrac{1}{2}$, & $\dfrac{1}{3\sqrt{5}} \neq 0  \hspace{0.1cm}$ when  $p(t) = \dfrac{1}{2}$, & No\\[2.5 ex]
 &    $1-\dfrac{2p(t)}{3}>\dfrac{2}{3}$   when $0< p(t) < \dfrac{1}{2}$, & $\dfrac{2p(t)}{3\sqrt{5}} \neq 0$ when $0< p(t) < \dfrac{1}{2}$ & \\[2.5 ex]
 & where  & where & \\[2.5 ex]
 & $p(t)= \exp \left[ -gt \right]\left( \cos(g\omega t)+\dfrac{\sin(g\omega t)}{\omega}\right) $ &  $p(t)= \exp \left[ -gt \right]\left( \cos(g\omega t)+\dfrac{\sin(g\omega t)}{\omega}\right)$ & \\[2.5 ex]
 \hline

\end{tabular}
\end{center}

\end{widetext}

\begin{thebibliography}{99}

\bibitem{teleb} C. H. Bennett, G. Brassard, C. Crepeau, R. Jozsa, A. Peres, and W. K. Wootters, \href{https://journals.aps.org/prl/abstract/10.1103/PhysRevLett.70.1895}{Phys. Rev. Lett. {\bf 70}, 1895 (1993)}.

\bibitem{teleapp1} H.-J. Briegel, W. Dur, J. I. Cirac, and P. Zoller,  \href{https://journals.aps.org/prl/abstract/10.1103/PhysRevLett.81.5932}{Phys. Rev. Lett. {\bf 81}, 5932 (1998)}.

\bibitem{teleapp2} D. Gottesman, and I. L. Chuang,  \href{https://www.nature.com/articles/46503}{Nature {\bf 402}, 390 (1999)}.

\bibitem{teleapp3} R. Raussendorf, and H.-J. Briegel,  \href{https://journals.aps.org/prl/abstract/10.1103/PhysRevLett.86.5188}{Phys. Rev. Lett. {\bf 86}, 5188
(2001).}

\bibitem{mqt} W. Dur, and J. I. Cirac,  \href{https://www.tandfonline.com/doi/abs/10.1080/09500340008244039}{J. Mod. Opt. {\bf 47}, 247 (2000)}.

\bibitem{cqt} S. L. Braunstein, and H. J. Kimble,  \href{https://journals.aps.org/prl/abstract/10.1103/PhysRevLett.80.869}{Phys. Rev. Lett. {\bf 80}, 869 (1998)}.

\bibitem{exp1} D. Boschi, S. Branca, F. De Martini, L. Hardy, and S. Popescu,   \href{https://journals.aps.org/prl/abstract/10.1103/PhysRevLett.80.1121}{Phys. Rev. Lett. {\bf 80}, 1121 (1998)}.

\bibitem{exp2} R. Ursin,  T. Jennewein, M. Aspelmeyer, R. Kaltenbaek, M. Lindenthal, P. Walther, and  A. Zeilinger,  \href{https://www.nature.com/articles/430849a}{Nature {\bf 430}, 849 (2004)}. 

\bibitem{exp3} C. Nolleke, A. Neuzner, A. Reiserer, C. Hahn, G. Rempe, and S. Ritter,  \href{https://journals.aps.org/prl/abstract/10.1103/PhysRevLett.110.140403}{Phys. Rev. Lett. {\bf 110}, 140403 (2013)}.

\bibitem{exp4} X.-S. Ma, T. Herbst, T. Scheidl, D. Wang, S. Kropatschek, W. Naylor, B. Wittmann, A. Mech, J. Kofler, E. Anisimova, V. Makarov, T. Jennewein, R. Ursin, and A. Zeilinger,  \href{https://www.nature.com/articles/nature11472}{Nature {\bf 489}, 269273 (2012)}.

\bibitem{expgs} J.-G. Ren, P. Xu, H.-L. Yong, L. Zhang, S.-K. Liao, J. Yin, W.-Y. Liu, W.-Q. Cai, M. Yang, L. Li, K.-X. Yang, X. Han, Y.-Q. Yao, J. Li, H.-Y. Wu, S. Wan, L. Liu, D.-Q. Liu, Y.-W. Kuang, Z.-P. He, P. Shang, C. Guo, R.-H. Zheng, K. Tian, Z.-C. Zhu, N.-L. Liu, C.-Y. Lu, R. Shu, Y.-A. Chen, C.-Z. Peng, J.-Y. Wang, and J.-W. Pan,  \href{https://www.nature.com/articles/nature23675}{Nature {\bf 549}, 70 (2017)}.

\bibitem{tele} R. Horodecki, M. Horodecki, and P. Horodecki,   \href{https://www.sciencedirect.com/science/article/abs/pii/0375960196006391?via}{Phys. Lett. A {\bf 222}, 21 (1996)}.

\bibitem{can1} M. Horodecki, P. Horodecki, and R. Horodecki,  \href{https://journals.aps.org/pra/abstract/10.1103/PhysRevA.60.1888}{Phys. Rev. A {\bf 60}, 1888 (1999)}.

\bibitem{f1} N. Gisin,   \href{https://www.sciencedirect.com/science/article/abs/pii/S0375960196800028?via}{Phys. Lett. A {\bf 210}, 157 (1996)}.

\bibitem{f2} F. Verstraete, and H. Verschelde,   \href{https://journals.aps.org/prl/abstract/10.1103/PhysRevLett.90.097901}{Phys. Rev. Lett. {\bf 90}, 097901 (2003)}.

\bibitem{fd1} J. Bang, J. Ryu, and D. Kaszlikowski,   \href{https://iopscience.iop.org/article/10.1088/1751-8121/aaac35}{J. Phys. A: Math. Theor.  {\bf 51}, 135302 (2018)}.

\bibitem{FD} A. Ghosal, D. Das, S. Roy, and S. Bandyopadhyay,  \href{https://iopscience.iop.org/article/10.1088/1751-8121/ab6ede}{J. Phys. A: Math. Theor. {\bf 53}, 145304 (2020)}.

\bibitem{Opt} A. Ghosal, D. Das, S. Roy, and S. Bandyopadhyay,   \href{https://journals.aps.org/pra/abstract/10.1103/PhysRevA.101.012304}{Phys. Rev. A \textbf{101}, 012304 (2020)}.

\bibitem{fd4} S. Roy, and A. Ghosal,   \href{https://journals.aps.org/pra/abstract/10.1103/PhysRevA.102.012428}{Phys. Rev. A {\bf 102}, 012428 (2020)}.


\bibitem{extra1} S. Roy, S. Mal, and A. Sen De,   \href{https://arxiv.org/abs/2010.14552}{arXiv:2010.14552 [quant-ph].}

\bibitem{extra2} P. Bej, S. Halder, and R. Sengupta,   \href{https://arxiv.org/abs/2102.01022}{arXiv:2102.01022 [quant-ph].}


\bibitem{gate.dist} L. Hjortshøj Pedersen, N. M. Møller and K. Mølmer, \href{https://www.sciencedirect.com/science/article/pii/S0375960108015041}{Phys. Lett. A \textbf{372}, 7028-7032 (2008).}


\bibitem{cite1} J. Bang, S.W. Lee, H. Jeong, and J. Lee,   \href{https://journals.aps.org/pra/abstract/10.1103/PhysRevA.86.062317}{Phys. Rev. A \textbf{86}, 062317 (2012)}.

\bibitem{breuerpet} H.-P. Breuer, F. Petruccione \textit{et al.}, {\it The theory of open quantum systems} (Oxford University Press on Demand,2002).

\bibitem{sbbook} S. Banerjee, {\it Open Quantum Systems: Dynamics of Non-classical Evolution}, Vol. 20 (Springer, 2018).

\bibitem{kraus} K. Kraus,   \href{https://www.sciencedirect.com/science/article/abs/pii/0003491671901084}{Ann. Physics {\bf 64}, 311–335 (1971)}.

\bibitem{kraus2} K. Kraus, {\it States, Effects and Operations: Fundamental Notions of Quantum Theory}, Springer-Verlag, 1983.

\bibitem{choi} M.-D. Choi,   \href{https://www.sciencedirect.com/science/article/pii/0024379575900750}{Linear Algebra and Its Applications, {\bf 10}, 285-290, (1975)}.


\bibitem{conc} W. K. Wootters,   \href{https://journals.aps.org/prl/abstract/10.1103/PhysRevLett.80.2245}{Phys. Rev. Lett. {\bf 80}, 2245 (1998)}.


\bibitem{hs1} R. Horodecki, and P. Horodecki,   \href{https://www.sciencedirect.com/science/article/abs/pii/0375960195009051?via}{Phys. Lett. A {\bf 210}, 227 (1996)}.

\bibitem{hs2} R. Horodecki, and M. Horodecki,   \href{https://journals.aps.org/pra/abstract/10.1103/PhysRevA.54.1838}{Phys. Rev. A {\bf 54}, 1838 (1996)}.

\bibitem{can2} P. Badziag, M. Horodecki, P. Horodecki, and R. Horodecki,   \href{https://journals.aps.org/pra/abstract/10.1103/PhysRevA.62.012311}{Phys. Rev. A {\bf 62}, 012311 (2000)}.


\bibitem{QC} F. Verstraete, and H. Verschelde,   	\href{https://arxiv.org/abs/quant-ph/0202124}{arXiv:quant-ph/0202124}.


\bibitem{omkarsinglequbit} S. Omkar, R. Srikanth, and S. Banerjee,   \href{https://link.springer.com/article/10.1007/s11128-013-0628-3}{Quantum Inf Process, {\bf 12}, 3725 (2013)}.

\bibitem{ruskai} M. B. Ruskai, S. Szarek, and E. Werner,   \href{https://www.sciencedirect.com/science/article/pii/S002437950100547X?via}{Linear Algebra Appl. {\bf 347}, 159 (2002)}.

\bibitem{jam} A. Jamiolkowski,   \href{https://www.sciencedirect.com/science/article/pii/0034487772900110}{Reports on Mathematical Physics {\bf  3}, 275-278 (1972)}.

\bibitem{choi2} M. Jiang, S. Luo, and S. Fu,   \href{https://journals.aps.org/pra/abstract/10.1103/PhysRevA.87.022310}{Phys. Rev. A {\bf 87}, 022310 (2013)}.

\bibitem{QC2} A. Fujiwara, and P. Algoet,   \href{https://journals.aps.org/pra/abstract/10.1103/PhysRevA.59.3290}{Phys. Rev. A {\bf 59}, 3290 (1999)}.

\bibitem{QC3} C. King, and M.-B. Ruskai,   \href{https://arxiv.org/abs/quant-ph/9911079}{IEEE Trans. on Inf. Theory,
{\bf 47}, 192–209 (2001)}.

\bibitem{oneshot} R. Pal, S. Bandyopadhyay, and S. Ghosh,   \href{https://link.aps.org/doi/10.1103/PhysRevA.90.052304}{Phys. Rev. A \textbf{90}, 052304 (2014)}.

\bibitem{GADC} S. Khatri, K. Sharma, and M. M. Wilde,   \href{https://link.aps.org/doi/10.1103/PhysRevA.102.012401}{Phys. Rev. A \textbf{102}, 012401 (2020)}.

\bibitem{SGAD} R. Srikanth, and S. Banerjee,   \href{https://link.aps.org/doi.org/10.1103/PhysRevA.77.012318}{Phys. Rev. A 77, 012318 (2008)}.

\bibitem{uppbfidelity} F. Verstraete and H. Verschelde,   \href{https://link.aps.org/doi/10.1103/PhysRevA.66.022307}{Phys. Rev. A \textbf{66}, 022307 (2002)}.





\bibitem{x1} T. Yu, and J. H. Eberly   Quantum Inf. Comput. {\bf 7}, 459 (2007).

\bibitem{x2} P. E. M. F. Mendonca, M. A. Marchiolli, and D. Galetti,   \href{https://www.sciencedirect.com/science/article/abs/pii/S000349161400253X?via}{Annals of Physics {\bf 351}, 79 (2014).}


\bibitem{rice}  S. O. Rice. {\it Stochastic Processes in Physics and Chemistry} (Elsevier, Amsterdam, 1992).

\bibitem{pradeep} N. P. Kumar, S. Banerjee, R. Srikanth, V. Jagadish, and F. Petruccione,  \href{https://doi.org/10.1142/S1230161218500142}{Open Systems and Information Dynamics {\bf 25}, 1850014 (2018)}.

\bibitem{piilo} L. Mazzola, J. Piilo, and S. Maniscalco,   \href{https://www.worldscientific.com/doi/abs/10.1142/S021974991100754X}{Int. J. Quant. Info, {\bf 9}, 981–991 (2011)}.

\bibitem{ali}  M. Ali,   \href{https://doi.org/10.1016/j.physleta.2014.05.030}{Phys. Lett. A, {\bf 378}, 2048–2053 (2014)}.

\bibitem{pinto}  J. P. G. Pinto, G. Karpat, and F. F Fanchini,   \href{https://journals.aps.org/pra/abstract/10.1103/PhysRevA.88.034304}{Phys. Rev. A, {\bf 88}, 034304 (2013)}.

\bibitem{control} S. K. Goyal, S. Banerjee, and S. Ghosh,   \href{https://journals.aps.org/pra/abstract/10.1103/PhysRevA.85.012327}{Phys. Rev. A, {\bf 85}, 012327 (2012)}.

\bibitem{ornstein} G. E. Uhlenbeck, and L. S. Ornstein,   \href{https://journals.aps.org/pr/abstract/10.1103/PhysRev.36.823}{Phys. Rev, {\bf 36}, 823 (1930)}.

\bibitem{yueberley}  T. Yu, and J. H. Eberly,   \href{https://doi.org/10.1016/j.optcom.2009.10.042}{Opt. Comm, {\bf 283}, 676–680 (2010)}.

\bibitem{falci}  E. Paladino, Y. M. Galperin, G. Falci, and B. L. Altshuler,   \href{https://journals.aps.org/rmp/abstract/10.1103/RevModPhys.86.361}{Rev. Mod. Phys, {\bf 86}, 361 (2014)}.

\bibitem{weakerthanCP} S. Utagi, R. Srikanth, and S. Banerjee,   \href{https://www.nature.com/articles/s41598-020-72211-3}{Scientific Reports
{\bf 10}, 15049 (2020)}.

\bibitem{wodco} S. Daffer, K. Wodkiewicz,  J. D. Cresser,  and J. K.McIver,   \href{https://journals.aps.org/pra/abstract/10.1103/PhysRevA.70.010304}{Phys. Rev. A {\bf 70}, 010304(R) (2004)}.

\bibitem{banghosh}   S. Banerjee, and R. Ghosh,   \href{https://iopscience.iop.org/article/10.1088/1751-8113/40/45/014}{J. Phys. A: Math. Theor.  {\bf 40}, 13735 (2007)}.

\bibitem{NMDeph}   U. Shrikant, R. Srikanth, and S. Banerjee,   \href{https://journals.aps.org/pra/abstract/10.1103/PhysRevA.98.032328}{Phys. Rev. A {\bf 98}, 032328 (2018)}.

\bibitem{bansrik}  S.  Omkar,  S.  Banerjee,  R.  Srikanth,  and  A.  K.  Alok,  \href{https://arxiv.org/abs/1408.1477}{Quantum Inf. and Comp. {\bf 16}, 0757 (2016)}.



\end{thebibliography}
\end{document}